\documentclass[12pt, a4paper]{amsart}

\usepackage{xcolor}
\definecolor{MyBlue}{cmyk}{1,0.13,0,0.63}
\definecolor{MyGreen}{cmyk}{0.91,0,0.88,0.52}
\newcommand{\mylinkcolor}{MyBlue}
\newcommand{\mycitecolor}{MyGreen}
\newcommand{\myurlcolor}{black}

\usepackage{hyperref}
\hypersetup{%
  bookmarksnumbered=true,bookmarksopen=false,%
  plainpages=false,
  linktocpage=true,%
  colorlinks=true,breaklinks=true,%
  linkcolor=\mylinkcolor,citecolor=\mycitecolor,urlcolor=\myurlcolor,%
  pdfpagelayout=OneColumn,%
  pageanchor=true,%
}

\usepackage{graphicx}
\usepackage{amsmath, amsthm,  amsfonts, amscd}
\usepackage{amssymb}
\usepackage{url}
\usepackage{fullpage}
\usepackage{mathtools}
\usepackage[all]{xy}
\usepackage{slashed}
\usepackage{multirow}

\newtheorem{thm}{Theorem}[section]
\newtheorem*{thm*}{Theorem}
\newtheorem{cor}[thm]{Corollary}
\newtheorem{lemma}[thm]{Lemma}
\newtheorem{prop}[thm]{Proposition}

\theoremstyle{definition}
\newtheorem{defn}[thm]{Definition}

\theoremstyle{remark}
\newtheorem{remark}[thm]{Remark}

\newcommand{\D}{\ensuremath{\mathcal{D}}}

\newcommand{\End}{\ensuremath{\mathrm{End}}}
\newcommand{\Hom}{\ensuremath{\mathrm{Hom}}}

\newcommand{\wt}{\ensuremath{\widetilde}}

\newcommand{\R}{\ensuremath{\mathbb{R}}}
\newcommand{\N}{\ensuremath{\mathbb{N}}}
\newcommand{\Z}{\ensuremath{\mathbb{Z}}}

\newcommand{\C}{\ensuremath{\mathbb{C}}}
\newcommand{\T}{\ensuremath{\mathbb{T}}}

\def\calT{\mathcal{T}}
\def\calC{\mathcal{C}}

\def\calK{\mathcal{K}}
\def\calB{\mathcal{B}}
\def\calH{\mathcal{H}}

\def\calA{\mathcal{A}}
\def\calN{\mathcal{N}}

\def\calE{\mathcal{E}}

\def\calU{\mathcal{U}}

\def\bP{\mathbf{P}}

\newcommand{\ol}{\overline}

\theoremstyle{definition}

\DeclareMathOperator{\Dom}{Dom}

\DeclareMathOperator{\Index}{Index}

\DeclareMathOperator{\Ker}{Ker}

\DeclareMathOperator{\Tr}{Tr}

\newcommand{\A}{\mathcal{A}}

\newcommand{\rst}[1]{\ensuremath{{\mathbin\upharpoonright}%
\raise-.5ex\hbox{$#1$}}}

\newcommand{\mat}[2]{\left(\!\!\begin{array}{#1}#2\end{array}\!\!\right)}

\makeatletter

\newcommand{\Rmnum}[1]{\expandafter\@slowromancap\romannumeral #1@}
\makeatother

\newcommand{\mywedge}{\bigwedge\nolimits^{\!\ast}}
\newcommand{\hS}{\tilde S}

\newcommand{\ext}{\mathrm{ext}}

\author{C. Bourne}
\address{Mathematical Sciences Institute, Australian National University, Canberra, ACT 0200, Australia; 
{\em and} School of Mathematics and Applied Statistics, University of Wollongong, 
Wollongong, NSW 2522, Australia; 
{\em and} Department Mathematik, Friedrich-Alexander-Universit\"{a}t Erlangen-N\"{u}rnberg, 
Cauerstra{\ss}e 11, 91058 Erlangen, Germany
{\em and}
Advanced Institute for Materials Research, Tohoku University,
2-1-1 Katahira, Aoba-ku, Sendai, 980-8577, Japan}
\email{christopher.j.bourne@gmail.com}

\author{J. Kellendonk}
\address{Univerisit\'{e} de Lyon, Universit\'{e} Claude Bernard Lyon 1, 
Institute Camille Jordan, CNRS UMR 5208, 69622 Villeurbanne, France}
\email{kellendonk@math.univ-lyon1.fr}

\author{A. Rennie}
\address{School of Mathematics and Applied Statistics, University of Wollongong, 
Wollongong, NSW 2522, Australia}
\email{renniea@uow.edu.au}

\date{\today}

\begin{document}

\begin{abstract}
We study the application of Kasparov theory to topological insulator 
systems and the bulk-edge correspondence. We consider observable algebras 
as modelled by crossed products, where bulk and edge systems may be linked 
by a short exact sequence. We construct unbounded Kasparov modules encoding the 
dynamics of the crossed product. We then link bulk and edge Kasparov modules 
using the Kasparov product. Because of the anti-linear symmetries that 
occur in topological insulator models, real $C^*$-algebras and $KKO$-theory must be used.
\end{abstract}

\title{The $K$-theoretic bulk-edge correspondence for topological insulators}
\maketitle

Keywords: Topological insulators, $KK$-theory, bulk-edge correspondence 

Subject classification: Primary 81R60, Secondary 19K35

\tableofcontents

\section{Introduction}
The bulk-edge correspondence is fundamental for 
topological insulators. Indeed, insulators in non-trivial 
topological phases are characterised by the existence of 
edge states at the Fermi energy which are robust against 
perturbations coming from disorder. This is a topological 
bulk-boundary effect, as explained in all details for complex 
topological insulators  in \cite{PSBbook} (see Corollary~6.5.5). 
The theory of~\cite{PSBbook} is based on the noncommutative 
topology ($K$-theory, cyclic cohomology and index theory) 
of complex $C^*$-algebras, an approach to solid state 
systems proposed and developed by Bellissard~\cite{BelGapLabel,Bellissard94}.

Complex topological insulators are topological insulators 
which are not invariant under a symmetry  
implemented by anti-linear operators, such as time reversal symmetry. 
The class of insulators with anti-linear symmetries require noncommutative topology 
for real algebras and here the bulk-edge correspondence is still in 
development. 

The bulk-edge correspondence in its noncommutative formulation 
has two sides. One side is the $K$-theoretic correspondence, 
where the boundary map of $K$-theory yields an equality 
between $K$-group elements, one of the $K$-group of the 
observable (bulk) algebra, and the other of the $K$-group of the edge algebra. 
The other side is ``dual to $K$-theory", namely a correspondence 
between elements of a theory which provides us with functionals 
on $K$-theory. The pairing of $K$-group elements with the dual theory
can be used to obtain 
topological numbers from the $K$-theory elements and, 
possibly, with a physical explanation of these numbers as 
topologically quantised entities (non-dissipative transport coefficients). 

In the complex case this has been achieved using cyclic cohomology~\cite{SBKR02,KSB04b, PSBbook}. 
The result of applying a functional coming from 
cyclic cohomology (one speaks of a pairing between cyclic cocycles 
and $K$-group elements) is a complex number and since the functional 
is additive the result is necessarily $0$ on elements of finite order. 
The most exciting topological invariants, however, lead to $K$-group 
elements which have order $2$, the Kane--Mele invariant being 
such an example. This is why we look into another theory dual to 
$K$-group theory, namely $K$-homology (its proper dual in the algebraic sense). 
It leads to functionals which applied to $K$-group elements are 
Clifford index valued and do not vanish on finite groups. 
{\em Our specific aim in this paper is to describe a computable 
version of the bulk-edge correspondence for $K$-homology.} 

We actually work in the more general setting of Kasparov's $KK$-theory 
of which $K$-theory and $K$-homology are special cases. 
In doing so we generalise the approach using Kasparov's theory 
 for the quantum Hall effect~\cite{BCR14} to the case of all topological insulators.

Whereas the specific details of the insulator feed into the 
construction of $K$-group elements, the $K$-homology class
appears to be stable for systems of the same dimension. We call it the 
fundamental $K$-cycle as it is 
constructed similarly to Kasparov's fundamental class 
for oriented manifolds~\cite{KasparovNovikov, LRV12}.
It may well be a key feature of these kinds of condensed matter models 
that the physics (alternatively the geometry and topology) is governed 
by classes of this type, which is effectively the 
fundamental class  of the momentum space.

An important aspect of the $K$-cycle is that it 
involves the Dirac operator in momentum space 
$\sum_{j=1}^d X_j \gamma^j$ with $X_j$ the components 
of the position operator in $d$ dimensions, and $\gamma^j$ the gamma 
matrices acting in a physical representation of the 
relevant algebra on $\ell^2(\Z^d,\C^\nu)$ or $\ell^2(\Z^{d-1},\C^\nu)$
(we work in the tight binding approximation). This operator, 
or rather its phase, has played a fundamental role ever 
since it was employed in the integer quantum Hall effect (the phase induces 
a singular gauge transformation sending one charge to 
infinity, a process which is at the heart of Laughlin's 
Gedankenexperiment). The associated index formula 
has recently been generalised to all topological insulators~\cite{GSB15}.

Computing boundary maps in $K$-homology is notoriously 
difficult. It comes down to realising $K$-homology groups 
of $A$ as Kasparov $KKO_i(A,\R)$-groups \cite{Kasparov80} 
and determining the Kasparov product with the so-called 
extension class. Until recently this seemed impossible in general. 
But recent progress on the formulation of $KK$-theory 
via unbounded Kasparov cycles 
\cite{BJ83,BMvS13,FR15,GMR,KL13,Kucerovsky97,MeslandMonster, MR15}
make it feasible in our case. We determine 
 an unbounded Kasparov cycle for the  
class of the Toeplitz extension 
of the observable algebra of the insulator. 
The task is then to compute the Kasparov product of 
the $KK$-class of this extension, denoted $[\mathrm{ext}]$, with the class of the fundamental 
$K$-cycle of the edge algebra, denoted $\lambda_e$. 
One of the important results of this paper is that the 
product is, up to sign, the $KK$-class of the bulk fundamental 
$K$-cycle $\lambda_b$, 
$$ 
[\mathrm{ext}]\hat\otimes_{A_e} [\lambda_e]=(-1)^{d-1}[\lambda_b],
$$ 
where $[\lambda_\bullet]$ denotes the class of the Kasparov module in 
$KK$-theory.
This is the dual side to the bulk-boundary correspondence, 
showing how to relate the fundamental $K$-cycles of the bulk and edge theories.  
We also emphasise that by 
working in the unbounded setting, our computations are explicit and have 
direct link to the underlying physics and geometry of the system. 

We can relate our results about fundamental $K$-cycles to 
topological phases by identifying the (bulk) invariants of 
topological phases as a pairing/product of the $K$-cycle 
$\lambda_b$ with a $K$-theory class $[x_b]$. 
We do not prescribe what the $K$-group element $[x_b]$ 
of the bulk has to be. It can be the homotopy class of 
a symmetry compatible Hamiltonian 
(translated from van Daele $K$-theory to $KK$-theory) or the class 
of the symmetry of the insulator constructed in \cite{BCR15}.
Leaving this flexible gives us the possibility to 
consider insulators systems of quite general symmetry type  
without affecting the central correspondence. It also allows
our approach to be adapted to different experimental arrangements, 
an important feature given the difficulties of measuring $\Z_2$-labelled
phases experimentally.

The $K$-theory side of the bulk-edge correspondence can be obtained 
by realising the $K$-groups of the bulk algebra $A$ as 
Kasparov $KKO_j(\R,A)$-groups so that the boundary 
map applied to a $KK$-class $[x_b]$ of $A$ is given by the 
Kasparov product $[x_b]\hat\otimes_A [\rm ext]$. The 
bulk-boundary correspondence between topological 
quantised numbers is then a direct consequence of the associativity of the Kasparov product  
$$ 
([x_b]\hat\otimes_A [\mathrm{ext}])\hat\otimes_{A_e} [\lambda_e]
=[x_b]\hat\otimes_A ([\mathrm{ext}]\hat\otimes_{A_e} [\lambda_e]).
$$
The equation says that the two topological quantised entities, 
that for the bulk $ [x_b]\hat\otimes_A [\lambda_b]$, 
and that for the edge $[x_e]\hat\otimes_{A_e} [\lambda_e]$, where 
$[x_e]=[x_b]\hat\otimes_A[\mathrm{ext}]$ corresponds to the image of the boundary map 
on $[x_b]$, are equal. Having said that, the result of the 
pairing lies in $KKO_{i+j+1}(\R,\R)$ which is a group 
generated by one element and not a number, and so 
still needs interpretation. While this number can sometimes be 
interpreted as an index (modulo $2$ valued, in certain 
degrees) we lack a better understanding,  
which had been possible in the complex case due to the 
use of derivations which make up the cyclic cohomology classes.

Apart from
\cite{BCR14}, our work relies most substantially on \cite{BCR15} which 
used real Kasparov theory to derive the groups that appear in the
`periodic table of topological insulators and superconductors' as outlined by Kitaev~\cite{Kitaev09}.
Our work also builds on and complements that in~\cite{FM13} 
for the commutative case and~\cite{GSB15,Kellendonk15, Kubota15b, Thiang14} 
for the noncommutative approach. 
We consider systems with weak disorder (that is, disordered Hamiltonians retaining a 
spectral gap). The substantial problem of strong disorder and localisation will not 
be treated here.

\subsection{Relation to other work}
There have been several mathematical papers detailing aspects 
of the bulk-edge correspondence for topological insulators with 
anti-linear symmetries. Graf and Porta prove a bulk-edge correspondence for two-dimensional 
Hamiltonians with odd time-reversal symmetry using Bloch bundles and without 
reference to $K$-theory~\cite{GP13}. Similar results are obtained by Avilla, Schulz-Baldes and 
Villegas-Blas using spin Chern numbers and an argument involving transfer matrices 
\cite{ASV13, SchulzBaldes13}. Spin Chern numbers are related to the noncommutative 
approach to the quantum Hall effect and, as such, allow samples with disorder to be 
considered. 

An alternative approach is taken by Loring, who derives the invariants of topological 
phases by considering almost commuting Hermitian matrices and their Clifford 
pseudospectrum~\cite{Loring15}. Such a viewpoint gives expressions for the invariants 
of interest that are 
amenable to numerical simulation. Loring also relates indices associated to $d$ and $d+1$ 
dimensional systems, a bulk-edge correspondence~\cite[Section 7]{Loring15}. What is 
less clear is the link between Loring's results and earlier results on the 
bulk-edge correspondence for the quantum Hall effect, particularly~\cite{SBKR02, KSB04b}.

Papers by Mathai and Thiang establish 
a $K$-theoretic bulk-edge correspondence for time-reversal 
symmetric systems~\cite{MT15b, MT15c}, and more recently with Hannabuss \cite{HMT} consider
the general case of topological insulators. Mathai and Thiang show that 
the invariants of interest in $K$-theory pass from bulk to edge under the 
boundary map of the real or complex Pimsner--Voiculescu sequence. Mathai and Thiang also 
show that, under T-duality, the boundary map in $K$-theory of tori can be expressed 
as the conceptually simpler restriction map. Similar results also appear in the 
work of Li, Kaufmann and Wehefritz-Kaufmann, who consider time-reversal symmetric 
systems and their relation to the $KO$, $KR$ and $KQ$ groups of $\T^d$~\cite{LKW15}. 
A bulk-edge correspondence then links the topological $K$-groups associated to the bulk and 
boundary using the Baum--Connes isomorphism and Poincar\'{e} duality.

Recent work by Kubota establishes a bulk-edge correspondence in $K$-theory for 
topological phases of quite general type~\cite{Kubota15b}. Kubota follows the 
general framework of~\cite{FM13, Thiang14} and considers twisted equivariant 
$K$-groups of uniform Roe algebras, which can be computed using the 
coarse Mayer--Vietoris exact sequence and coarse Baum--Connes map. 
Classes in such groups are associated to 
gapped Hamiltonians compatible with a twisted symmetry group. An edge 
invariant is also defined and is shown to be isomorphic to the bulk class under 
the boundary map of the coarse Mayer--Vietoris exact sequence in $K$-theory 
\cite{Kubota15a, Kubota15b}. The use 
of Roe algebras and coarse geometry means that there is the potential for systems 
with impurities and uneven edges to be considered.

\subsection{Outline of the paper}
We begin with a short review of unbounded Kasparov theory in Section 
\ref{sec:Kasparov_prelim}. We particularly focus on real Kasparov 
theory as it is relatively understudied in the literature.

The central mathematical content of this paper is in Section \ref{sec:KKO_bulkedge}, 
where we construct fundamental $K$-cycles $\lambda^{(d)}$ for (possibly twisted) 
$\Z^d$-actions of unital $C^*$-algebras. We can then link 
actions of different order by the Pimsner--Voiculescu short exact 
sequence~\cite{PV80}, which we represent with an unbounded Kasparov 
module. Finally we use the unbounded product to show that the fundamental 
$K$-cycle $\lambda^{(d)}$ can be factorised into the product of the extension Kasparov 
module and $\lambda^{(d-1)}$ (up to a basis ordering of Clifford algebra 
elements, which may introduce a minus sign). These results are then applied to 
$C^*$-algebraic models of disordered or aperiodic media in Section \ref{sec:disorder_bulkedge}, 
where the (bulk) algebra of interest is the twisted crossed product 
$C(\Omega)\rtimes_{\alpha,\theta}\Z^d$.

In Section \ref{sec:K_theory_and_pairings} we relate our result to 
topological phases by pairing the $K$-cycle $\lambda^{(d)}$ with 
a $K$-theory class $[x_b]$ related to gapped Hamiltonians with 
time reversal and/or particle-hole and/or chiral symmetry. 
We also include a detailed discussion on the computation of the 
bulk and edge pairings, using both 
Clifford modules and the Atiyah--Bott--Shapiro construction 
as well as semifinite spectral triples and the semifinite local 
index formula. Limitations and open questions are also 
considered.

\subsubsection*{Acknowledgements}
All authors thank the Hausdorff Research Institute for Mathematics for support. 
CB and AR thank Alan Carey, Magnus Goffeng and Guo Chuan Thiang for useful discussions.  
CB and AR acknowledge the support of the Australian Research Council and CB is 
also supported by a Postdoctoral Fellowship of the Japan Society for the 
Promotion of Science.


\section{Preliminaries on real Kasparov theory} \label{sec:Kasparov_prelim}
For the convenience of the reader and to establish notation we
briefly summarise the results in real Kasparov theory of use to us for the bulk-edge 
correspondence. The reader may consult~\cite{Kasparov80, SchroderKTheory} 
or~\cite[Appendix A]{BCR15} for more information.

\subsection{$KKO$-groups}
The use of Clifford algebras to define higher $KK$-groups makes it important to
work with $\Z_2$-graded $C^*$-algebras and $\Z_2$-graded tensor 
products, $\hat\otimes$. A basic reference is 
\cite[{\S}14]{Blackadar}.
Given a real $C^*$-module $E_B$ (written $E$ when $B$ is clear) 
over a $\Z_2$-graded $C^*$-algebra $B$, 
we denote by $\End_B(E)$  the algebra of endomorphisms of 
$E$ which are adjointable with respect to the $B$-valued inner-product 
$(\cdot\mid\cdot)_B$ on $E_B$. The 
algebra of compact endomorphisms $\End_B^0(E)$ is generated
by the operators 
$\Theta_{e_1,e_2}$ for $e_1,e_2\in E$ such that for $e_3\in E$
$$  
  \Theta_{e_1,e_2}(e_3) = e_1\cdot (e_2\mid e_3)_B
$$
with $e\cdot b$ the (graded) right-action of $B$ on $E$.
\begin{defn}
Let $A$ and $B$ be $\Z_2$-graded $C^*$-algebras.
A real unbounded Kasparov module $(\calA, {}_\pi{E}_B, D)$ is a
$\Z_2$-graded real $C^*$-module ${E}_B$, a dense $*$-subalgebra $\A\subset A$ with graded 
real homomorphism $\pi:\calA\to \End_B({E})$, and an
unbounded, regular, odd and self-adjoint operator $D$ such that for all $a\in \calA$
\begin{align*}
  &[D,\pi(a)]_\pm \in \End_B({E}),   &&\pi(a)(1+D^2)^{-1/2}\in \End_B^0({E}).
\end{align*}
\end{defn}
Where unambiguous, we will omit the representation $\pi$ 
and write unbounded Kasparov modules as $(\calA, {E}_B, D)$.
The results of Baaj and Julg~\cite{BJ83} continue to hold for real 
Kasparov modules, so given an unbounded module 
$(\calA,{E}_B,D)$ we apply the bounded transformation 
$F_D = D(1+D^2)^{-1/2}$
to obtain the real (bounded) Kasparov module $(A,{E}_B,F_D)$, 
where $A$ is the $C^*$-closure of the dense subalgebra $\calA$.

One can define notions of unitary equivalence, 
homotopy and degenerate Kasparov modules 
in the real setting (see~\cite[\S{4}]{Kasparov80}). 
Hence we can define the group $KKO(A,B)$ as the 
equivalence classes of real (bounded) 
Kasparov modules modulo the equivalence relation 
generated by these relations.

Clifford algebras are used to define higher $KKO$-groups and 
encode periodicity. In the real setting, we define $C\ell_{p,q}$ 
to be the real span of the mutually anti-commuting generators 
$\gamma^1,\ldots,\gamma^p$ and $\rho^1,\ldots,\rho^q$ such that
$$   
(\gamma^i)^2 = 1, \quad (\gamma^i)^* = \gamma^i, \quad (\rho^i)^2=-1,
\quad(\rho^i)^*=-\rho^i.
$$

We now recall the relation between real $KK$-groups and real $K$-theory.
\begin{prop}[\cite{Kasparov80}, \S{6}, Theorem 3] 
\label{prop:real_Real_k_with_kasparov_equivalence}
For trivially graded, $\sigma$-unital real algebras $A$, 
there is an isomorphism $KKO(C\ell_{n,0},A) \cong KO_n(A)$.
\end{prop}

Each short exact sequence of real $C^*$-algebras $0\to B\to C\to A\to 0$
with ideal $B$ and quotient algebra $A$ gives rise 
to an element of the extension group $\mathrm{Ext}_\R(A,B)$ 
and this group is related to the real Kasparov $KK$-groups.
\begin{prop}[\cite{Kasparov80}, {\S}7] \label{prop:Ext_is_KK^1}
If $A$ and $B$ are separable real $C^*$-algebras, then
\begin{align*} 
  \mathrm{Ext}_\R^{-1}(A,B) &\cong KKO(A\hat\otimes C\ell_{0,1},B) \cong KKO(A, B\hat\otimes C\ell_{1,0}).
\end{align*}
\end{prop}

\subsection{The product} \label{subsec:Kas_prod_intro}
The generality of the constructions and proofs in~\cite{Kasparov80} 
mean that all the central results in complex $KK$-theory carry 
over into the real (and Real) setting. In particular, the intersection product
$$  
KKO(A,B) \times KKO(B,C) \to KKO(A,C) 
$$
is still a well-defined map and other important properties such as stability 
$$  
KKO(A\hat\otimes\calK(\calH),B)\cong KKO(A,B)  
$$
continue to hold, where $\calK(\calH)$ is the algebra of real compact 
operators on a separable real Hilbert space. 

Let $(\calA,E^1_B,D_1)$ and $(\calB,E^2_C,D_2)$ be unbounded real 
$A$-$B$ and $B$-$C$ Kasparov modules. We would like to take the 
product at the unbounded level. Naively one would like to use the formula
$$
(\calA,(E^1\hat\otimes_BE^2)_C,D_1\hat\otimes1+1\hat\otimes D_2),
$$
where $(E^1\hat\otimes_BE^2)_C$ is the $B$-balanced $\Z_2$-graded tensor product.
This does not make sense as $D_2$ is not $\calB$-linear (nor $B$-linear) and so
$1\hat\otimes D_2$ does not descend to the balanced tensor product.
Instead one needs to choose a connection on $E^1_B$ to correct the naive 
formula for $1\hat\otimes\D_2$. First define the `$D_2$ one forms' by
$$
\Omega^1_{D_2}(\calB)=\Big\{\sum_jb_j[D_2,c_j]\in 
\mathrm{End}_C(E^2)
\,:\,  b_j,\,c_j\in\calB\Big\}.
$$
Then a $D_2$ connection on $E^1$ is a choice of 
dense $\calB$ submodule $\calE^1\subset E^1$
and a linear map
$$
\nabla:\calE^1\to \calE^1\hat\otimes_\calB\,\Omega^1_{D_2}(\calB)\quad\mbox{such that}\quad
\nabla(eb)=\nabla(e)b+e\hat\otimes[D_2,b].
$$
Setting $m:\Omega^1_{D_2}(\calB)\hat\otimes E^2\to E^2$ 
to be $m(a[D_2,b]\hat\otimes f)=a[D_2,b]f$, we then define
$$
(1\hat\otimes_\nabla D_2)(e\hat\otimes f)
=(-1)^{|e|}e\hat\otimes D_2f+(1\hat\otimes m)(\nabla(e)\hat\otimes f),
$$
with $|e|$ the degree of $e$ in the $\Z_2$-graded module $E^1_B$.
A short but illuminating calculation shows that $1\hat\otimes_\nabla D_2$ is well-defined on the
balanced tensor product, and it is reasonable to hope that the formula
\begin{equation}
D:=D_1\hat\otimes 1+1\hat\otimes_\nabla D_2
\label{eq:constructive}
\end{equation}
would define an operator on $E^1\hat\otimes_BE^2$ 
such that $(\calA,(E^1\hat\otimes_BE^2)_C,D)$
represents the product class.

In fact this is true in very many cases as proved in \cite{KL13,LRV12,MeslandMonster,MR15}.
These papers provide very general settings where the formula \eqref{eq:constructive}
can be guaranteed to produce an unbounded Kasparov module representing the product.
These proofs, all in the complex case, proceed by showing that the conditions of Kucerovsky's
theorem are satisfied. We will not develop such a general framework, but concretely
construct potential representatives according to the
recipe in Equation \eqref{eq:constructive}, and
check Kucerovsky's conditions directly. Importantly, Kucerovsky's theorem is
valid in the real case.

To state Kucerovsky's conditions, we start by defining a creation operator. Given $e_1\in{E}_B^1$ 
and a $\ast$-homomorphism $\psi:B\to\End_C({E}^2)$, we let 
$T_{e_1}\in\Hom_C({E}^2, {E}^1\otimes_B {E}^2)$ be given by 
$T_{e_1}e_2 = e_1\hat\otimes e_2$. One can check that $T_{e_1}$ is 
adjointable with $T_{e_1}^*(f_1\hat\otimes e_2) = \psi((e_1|f_1)_B)e_2$. 

\begin{thm}[Kucerovsky's criteria~\cite{Kucerovsky97}, Theorem 13] 
\label{thm:Kucerovsky_criterion}
Let $(\calA,{}_{\pi_1}{E}^1_B,D_1)$ and $(\calB,{}_{\pi_2}{E}^2_C,D_2)$ 
be unbounded  Kasparov modules. Write 
${E} := {E}^1\hat\otimes_B{E}^2$. Suppose that 
$(\calA,{}_{\pi_1}{E}_C,D)$ is an unbounded Kasparov module such that
\begin{description}
\item[Connection condition] For all $e_1$ in a dense 
subspace of $\pi_1(A){E}^1$, the commutators 
$$
\left[\mat{cc}{D&0\\0&D_2\\},\mat{cc}{0&T_{e_1}\\T_{e_1}^*&0\\}\right]
$$
are bounded on $\Dom(D\oplus D_2)\subset{E}\oplus{E}^2$;
\item[Domain condition] $\Dom(D)\subset\Dom(D_1\hat\otimes1)$;
\item[Positivity condition] For all $e\in\Dom(D)$, 
$$
((D_1\hat\otimes1)e|De) + (De|(D_1\hat\otimes1)e) \geq K (e|e)
$$ 
for some $K\in\R$. 
\end{description}
Then the class of $(\calA,{}_{\pi_1}{E}_C,D)$ in $KK(A,C)$ represents the  Kasparov product. 
\end{thm}

In fact, if $\nabla$ satisfies an extra Hermiticity condition (which can always be achieved),
an operator of the form $D=D_1\hat\otimes1+1\hat\otimes_\nabla D_2$
always satisfies the domain condition and connection condition. Under mild hypotheses, 
the operator $1\hat\otimes_\nabla D_2$ will be self-adjoint~\cite[Theorem 3.17]{MR15}
and the sum will have locally compact resolvent (that is, $\pi_1(a)(1+D^2)^{-1/2}$ compact for 
$a\in \calA$)~\cite[Theorem 6.7]{KL13}. The
self-adjointness of the sum and the boundedness of commutators $[D,\pi_1(a)]$ needs to be
checked directly, as does the positivity condition. In our examples all these extra 
conditions are satisfied and so the task of checking that we have a spectral triple 
representing the product is relatively straightforward.

\subsection{Semifinite theory}
An unbounded $A$-$\C$ or $A$-$\R$ Kasparov module is precisely a complex or 
real spectral triple as defined by Connes. Complex spectral triples satisfying 
additional regularity properties have the advantage 
that the local index formula by Connes and Moscovici~\cite{CM95} gives computable expressions 
for the index pairing with $K$-theory, a special case of the Kasparov product
$$
  K_\ast(A) \times KK^\ast(A,\C) \to K_0(\C)\cong \Z.
$$

We would like to find computable expressions for more general Kasparov products. 
To do this, we generalise the definition of spectral triple using  
semifinite von Neumann algebras in place of the bounded operators on Hilbert space.

Let $\tau$ be a fixed faithful, normal, semifinite trace on a von Neumann algebra 
$\calN$. We let $\calK_\calN$ be the $\tau$-compact
operators in $\calN$ (that is, the norm closed ideal generated by the 
projections $P\in\calN$ with $\tau(P)<\infty$).
\begin{defn}[\cite{CPRS2}]
A semifinite spectral triple $(\calA,\calH,D)$ relative to $(\calN,\tau)$ is given by a $\Z_2$-graded 
Hilbert space $\calH$, 
a graded $\ast$-algebra $\calA\subset\calN$ with (graded) representation on 
$\calH$ and a densely defined odd 
unbounded self-adjoint operator $D$ affiliated to $\calN$ such that
\begin{enumerate}
  \item $[D,a]_\pm$ is well-defined on $\Dom(D)$ and extends to a bounded operator on $\calH$ for 
  all $a\in\calA$,
  \item $a(1+D^2)^{-1/2}\in\calK_\calN$ for all $a\in\calA$.
\end{enumerate}
A semifinite spectral triple $(\calA,\calH,D)$ is $QC^\infty$ if for 
$\partial = [(1+D^2)^{1/2},\cdot]$ and $b\in\calA\cup[D,\calA]$, $\partial^j(b)\in\calN$ 
for all $j\in\N$.

A unital semifinite spectral triple is $p$-summable if $(1+D^2)^{s/2}$ is $\tau$-trace-class 
for all $s>p$.
\end{defn}
If we take $\calN=\calB(\calH)$ and $\tau = \Tr$, then we recover the 
usual definition of a spectral triple.

\begin{thm}[\cite{KNR}] \label{thm:semifinite_to_KK}
Let $(\calA,\calH,D)$ be a complex semifinite spectral triple associated to 
$(\calN,\tau)$ with $\calA$ separable and $A$ its $C^*$-completion. Then $(\calA,\calH,D)$ 
determines a class in $KK(A,C)$ with $C$ a separable $C^*$-subalgebra of $\calK_\calN$.
\end{thm}

Semifinite spectral triples can also be paired with $K$-theory elements 
by the following composition
\begin{equation} \label{eq:complex_semifinite_index_pairing}
  K_0(A) \times KK(A, C) \to K_0(C) \xrightarrow{\tau_\ast} \R,
\end{equation}
with the class in $KK(A,C)$ coming from Theorem \ref{thm:semifinite_to_KK}. 
The range of the semifinite index pairing is a discrete subset of $\R$ and can 
potentially detect finer invariants than the usual index pairing. Importantly, 
the local index formula can be generalised to $QC^\infty$ and $p$-summable 
semifinite spectral triples for $p\geq 1$~\cite{CPRS2,CPRS3}. Hence the semifinite 
local index formula may be used to compute the map in Equation 
\eqref{eq:complex_semifinite_index_pairing}.

Suppose that $(\calA, E_B, D)$ is an unbounded Kasparov module 
and the algebra $B$ possesses a faithful semifinite 
norm-lower semicontinuous trace. Then one can often
construct a semifinite spectral triple
(see for example~\cite{PR06}). The passage from Kasparov module to 
semifinite spectral triple is advantageous as the algebra $B$ is 
usually more closely linked to the example or problem under consideration 
than the algebra $C$ from Theorem \ref{thm:semifinite_to_KK}.
 If a sufficiently regular (complex) semifinite 
spectral triple can be constructed, then we may use the semifinite 
local index formula to compute the map given in Equation 
\eqref{eq:complex_semifinite_index_pairing} (with $C$ replaced by $B$). 
Therefore semifinite 
spectral triples and index theory can be employed in order to 
compute pairings of complex $K$-theory 
classes with unbounded Kasparov modules.

\begin{remark}[Real semifinite spectral triples and the local index 
formula]
While the definition of a semifinite spectral triple and 
semifinite index pairing can be extended to real algebras, the semifinite 
local index formula can not be used to detect 
torsion invariants as it involves a mapping to cyclic cohomology. 
However, one may naturally ask whether the semifinite local 
index formula can be used to detect integer invariants in the real 
setting (e.g. arising from $K_0(\R)$ or $K_4(\R)$). We are of the opinion 
that a local formula can be employed to access real integer invariants as in \cite{CPRS2,CPRS3}, 
though the details of the proof given for the complex case need to be checked
for the real case.
\end{remark}


\section{Kasparov modules and boundary maps of twisted $\Z^d$-actions}
\label{sec:KKO_bulkedge}
Our motivation is to study topological invariants of disordered 
or aperiodic topological states in the tight-binding approximation.
Suppose $A$ is an observable algebra acting on $\ell^2(\Z^d)$ and is such that 
$A \cong C\rtimes_\beta \Z$ with 
$C$ acting on $\ell^2(\Z^{d-1})$. We consider the 
Toeplitz extension~\cite{PV80,SchroderKTheory}
$$
   0 \to C\otimes\calK(\ell^2(\N)) \to \calT(\beta) \to A \to 0.
$$
As in~\cite{SBKR02,PSBbook} and Section \ref{sec:disorder_algebra_and_ext}, 
the algebra $\calT(\beta)$ acts on $\ell^2(\Z^{d-1}\times\N)$ and can 
be interpreted as the algebra of observables on the space with boundary. The ideal 
$C\otimes \calK$ is thought of as the observables concentrated at the boundary $\ell^2(\Z^{d-1}\times \{0\})$ 
with compact decay into the interior of the sample. The quotient 
$A$ acts on a space without boundary (approximating the interior of a sample) 
and is often called the bulk algebra.
It follows from Kasparov's work that any short exact sequence of (trivially graded)
$C^*$-algebras gives 
rise to a long exact sequence in $K$-homology whose boundary map is given by 
the internal Kasparov product with the extension class defined by the Toeplitz extension. 
Hence topological properties of the algebra without boundary $A$ can be related 
to the edge algebra $C$ by 
the boundary map of $K$-theory or $K$-homology (or $KK$-theory more generally).

Boundary maps in $KK$-theory are in general very hard to compute. 
In our case we use the fact that for disordered media our bulk and edge 
algebras arise as (twisted) crossed products. Therefore 
we consider $A \cong B\rtimes_{\alpha,\theta}\Z^d$ 
and $C\cong B\rtimes_{\alpha^{\|},\theta}\Z^{d-1}$ for $B$ a separable and 
unital (real or complex) $C^*$-algebra. The dynamics of the 
crossed product allows us to build 
explicit unbounded Kasparov modules for the $A$ and $C$, which are then shown to be 
directly related under the boundary map in $KK$-theory.  Our formula is explicit and avoids 
homotopy arguments.

In Section \ref{sec:disorder_bulkedge} we will specialise to the physically interesting 
case of $B=C(\Omega)$, where $C(\Omega)\rtimes_{\alpha,\theta}\Z^d$ is the (bulk) observable 
algebra of a disordered or aperiodic solid state system~\cite{BelGapLabel, PSBbook}.

\subsection{Fundamental $K$-cycles for $\Z^d$-actions}
We construct an unbounded Kasparov module encoding a (twisted) $\Z^d$-action on a real or complex $C^*$-algebra $B$. 
We remark that similar constructions appear 
in~\cite{HawkinsThesis,PR06,Prodan15} in the complex setting.

Let $B$ be a separable and unital $C^*$-algebra with action $\alpha$ of $\Z^d$ 
and twisting cocycle $\theta:\Z^d\times\Z^d \to \calU(B)$, with $\calU(B)$ the unitaries of $B$. 
The twisted 
crossed product $A:=B\rtimes_{\alpha,\theta}\Z^d$ is the universal 
$C^*$-completion of the algebraic crossed product $\calA:=B_{\alpha,\theta}\Z^d$ 
given by finite sums $\sum_{n\in\Z^d} S^n b_n  $ where $b_n\in B$, 
$n\in\Z^d$ is a multi-index
and $S^n = S_1^{n_1}\cdots S_d^{n_d}$ is a product of powers of $d$ 
abstract unitary elements $S_i$ subject to the multiplication 
extending that of $B$ by
\begin{align*}
&S_i b = \alpha_{i}(b) S_i,    &&S_i S_j = \theta_{ij} S_jS_i,  &&S_i^* = S_i^{-1}.
\end{align*}
The map $\alpha_i$ is the automorphism 
corresponding to the action of $e_i\in\Z^d$ for $e_i$ the standard generators of $\Z^d$.
The elements 
$\theta_{ij}$ belong to $B$ and can be obtained from the cocycle $\theta$. 

As an example,
let us consider the case $d=1$ in which $\theta=1$.
Given a left action $\rho$ of $B$ on a module $M$ one 
obtains a left action $\pi$ of $ B_\alpha\Z$ on the 
module given by the algebraic tensor product $\ell^2(\Z)\otimes M$.
The action $\pi$ is given on elementary tensors by
\begin{align} \label{eq-rep}
 &\pi(b)(\delta_j\otimes \xi) =  \delta_j\otimes \rho(\alpha^{-j}(b))\xi, 
 &&\pi(S)(\delta_j\otimes \xi) = \delta_{j+1}\otimes \xi
\end{align}
with $\{\delta_j\}_{j\in\Z}$ the standard basis of $\ell^2(\Z)$ and $\xi\in M$.
We will be interested in the case where $M$ is a $C^*$-module so 
that we can expect the action to extend to an action of the $C^*$-crossed 
product $B\rtimes_\alpha\Z$ on the completion of the algebraic 
tensor product $\ell^2(\Z)\otimes M$ which we denote by $\ell^2(\Z,M)$.

The above procedure iterated for higher-order crossed products 
yields an action of 
$A=B\rtimes_{\alpha,\theta}\Z^d$ on $\ell^2(\Z^d,M)$ if $A$ can be 
rewritten as a $d$-fold iterated crossed product with $\Z$. This is always the 
case if 
we allow $B$ to be replaced by its stabilisation~\cite{PR89}, 
but the left-action may also be modified to yield directly a 
representation of $B\rtimes_{\alpha,\theta}\Z^d$~\cite{PR89}. 

We consider the left regular 
representation of $B$ on the $C^*$-module 
$B_B$ 
and obtain a representation of the crossed product
$B\rtimes_{\alpha,\theta} \Z^d$ on the $C^*$-module $\ell^2(\Z^d)\otimes B =: \ell^2(\Z^d,B)$. 
More precisely, the map 
$\imath : B_{\alpha,\theta}\Z^d \to \ell^2(\Z^d, B)$, $\imath( S^n b) = \delta_n\otimes b$ 
for $\{\delta_n\}_{n\in\Z^d}$ the canonical basis of $\ell^2(\Z^d)$
is an inclusion of the dense subalgebra $\calA$ into the  $C^*$-module such that the 
representation of $\calA=B_{\alpha,\theta}\Z$ 
on the image of $\imath$ is given by left multiplication: 
$$
  S^n b_1 \cdot (\delta_m\otimes b_2) = \imath(S^n b_1 S^m b_2) = \delta_{n+m} \otimes  
\alpha^{-m-n}(\theta(m,n)) \alpha^{-m}(b_1) b_2
 $$
with multi-index notation $\alpha^m = \alpha_1^{m_1}\cdots \alpha_d^{m_d}$.
The  $C^*$-module is the completion of $\imath(\calA)$ 
with respect to the $B$-valued inner product
$$
   \left( \lambda_1 \otimes b_1 \mid \lambda_2 \otimes b_2 \right)_B = 
     \langle \lambda_1, \lambda_2 \rangle_{\ell^2(\Z^d)} \, b_1^* b_2
$$
and norm $\| \xi\|^2 = \|(\xi\mid \xi)_B \|_B$. The $C^*$-module also carries 
a right-action of $B$ by right-multiplication. Note that on $\imath(\calA)$ the inner product may be written as
$$
   \left( \imath(a_1) \mid \imath(a_2) \right)_B = \Phi_0(a_1^*a_2), \qquad
   \Phi_0(S^n b) = \begin{cases} b, & n=0 \\ 0, & \text{otherwise} \end{cases}.
$$
As the conditional expectation $\Phi_0:A\to B$ is positive, $\Phi_0(b^*ab) \leq \|a\|_A\Phi_0(b^*b)$ 
for any positive $a\in A$ and $b\in B$.

\begin{prop} \label{prop:left_action_adjointable}
The left-action of $\calA$ on $\ell^2(\Z^d,B)$ extends to an adjointable 
representation of $A=B\rtimes_{\alpha,\theta} \Z^d$.
\end{prop}
\begin{proof}
We first check adjointability of the elements of $\A$ acting on $\imath(\calA)$ by computing
\begin{align*}
   \left( a\cdot \imath(a_1) \mid \imath(a_2) \right)_{B} 
     &= \Phi_0((aa_1)^*a_2) \\
     &= \Phi_0(a_1a^*a_2) =  \left( \imath(a_1) \mid  a^*\cdot\imath(a_2) \right)_{B} 
\end{align*}
Furthermore, the action of $\A$ on $\imath(\calA)$ is uniformly bounded:
\begin{align*}
\|a\|^2_{\End} &=
 \sup_{\substack{ a' \in \calA \\ \|\imath(a')\|=1}} \|(a\cdot \imath(a')\mid a\cdot \imath(a'))_{B}\| \\
&   \leq  \sup_{\substack{ a' \in \calA \\ \|\imath(a')\|=1}}
 \| a^*a\|_A\, \|\Phi_0({a'}^*a')\|
= \|a^*a\|_A.
\end{align*}
Hence the action extends by continuity, first from $\calA$ on $\imath(\calA)$
to a bounded action of $\calA$ on $\ell^2(\Z^d,B)$ and then to a bounded action of the whole $C^*$-algebra $A$ on $\ell^2(\Z^d,B)$. It then follows that all elements of $A$ are adjointable.
\end{proof}

Inside $\ell^2(\Z^d,B)$ we consider the elements $e_m=\delta_m \otimes 1_B$ 
for $m\in\Z^d$. 
We compute 
that for $b\in B$
\begin{align*}
 \Theta_{e_l,e_m}\,e_nb =  \Theta_{e_l,e_m} (\delta_n \otimes b) 
   &= (\delta_l \otimes 1_B)\cdot
    \left(\delta_m\otimes 1_B \mid \delta_n \otimes b \right)_B  \\
    &= \delta_{m,n} \, \delta_l \otimes b=\delta_{m,n}\,e_l\cdot b.
\end{align*}
In particular we note that 
$$ \sum_{m\in\Z^d} \Theta_{e_m, e_m} = \mathrm{Id}_{\ell^2(\Z^d,B)} $$
and so $\{{e_m}\}_{m\in\Z^d}$ is a frame for the module $\ell^2(\Z^d,B)$.

The last ingredient we need for a Kasparov module is a Dirac-like operator, which 
we construct using the position operators, 
$X_j: \Dom(X_j) \to \ell^2(\Z^d,B)$ for 
$j\in\{1,\ldots,d\}$ such that $X_j(\delta_m\otimes b) = m_j (\delta_m \otimes b)$ for any 
$m\in\Z^d$. We construct the Dirac-like operator via an explicit 
Clifford action.
On the graded vector space $\mywedge\R^{d}$ (we denote the grading by  $\gamma_{\bigwedge^* \R^d} $)
there is a representation of $C\ell_{d,0}$ and a representation of $C\ell_{0,d}$. The
generators $\gamma^j$ of $C\ell_{d,0}$ and the generators $\rho^j$ of $C\ell_{0,d}$ 
act by
\begin{align*}
&\gamma^j(w) = e_j \wedge w + \iota(e_j)w,
  &&\rho^j(w) = e_j\wedge w - \iota(e_j)w, 
\end{align*}
for $\{e_j\}_{j=1}^{d}$ the standard basis of $\R^{d}$,
$w\in\bigwedge^*\R^{d}$ and $\iota(v)w$ the contraction of $w$ along $v$. 
These two actions graded-commute.
On the tensor product space
$\ell^2(\Z^d,B)\otimes \mywedge\R^{d}$ we define
$$
D := \sum_{j=1}^{d} X_j\otimes \gamma^j.
$$
A simple check shows that $D$ is odd, self-adjoint and regular on 
$\ell^2(\Z^d,B)\otimes \bigwedge^*\R^d$.

\begin{prop} 
\label{prop:general-toeplitz}
Consider a possibly twisted $\Z^d$-action $\alpha,\theta$ on a 
separable and unital $C^*$-algebra 
$B$. Let $A$ be the associated crossed product with dense subalgebra 
$\calA = B_{\alpha,\theta}\Z^d$.
The data
\begin{equation*} \label{eq:spectral_trip} 
 \lambda^{(d)} = \bigg( \calA \hat\otimes C\ell_{0,d},\, \ell^2(\Z^{d},B)_B \otimes 
  \bigwedge\nolimits^{\!\ast}\R^{d},\, \sum_{j=1}^{d} X_j\otimes \gamma^j, \,
  \gamma_{\bigwedge^* \R^d}  \bigg)
\end{equation*}
defines an unbounded $A\hat\otimes C\ell_{0,d}$-$B$ Kasparov module. 
The $C\ell_{0,d}$-action is generated by the operators $\rho^j$.
In the complex case we have $\C$ in place of $\R$ in the above formula.
\end{prop}
\begin{proof}
The left-action of $A$ is adjointable by Proposition 
\ref{prop:left_action_adjointable}. By construction $D$ 
graded-commutes with the left Clifford representation. What 
remains to be checked is that $[D,a] = \sum_j [X_j,a]\otimes \gamma^j$ 
is adjointable for $a\in\calA$ and $(1+D^2)^{-1/2}$ is compact.
It is directly verified that
$$
[X_j,S^nb] = n_j S^nb
$$ 
and so we see that $[X_j,a]\in \calA$ for all $a\in\calA$. 
In particular the commutator is adjointable. Furthermore, 
$D^2 = \sum_j X_j^2\otimes 1_{\bigwedge^*\R^d}$ and hence $(1+D^2)^{-1/2} = (1+|X|^2)^{-1/2}\otimes 1_{\bigwedge^*\R^d}$,
where $|X|^2 = \sum_j X_j^2$. Using the frame $\{e_m\}_{m\in\Z^d}$, we
note that $(1+|X|^2)e_m = (1+|m|^2)e_m$ and so
$$ 
(1+D^2)^{-\frac12} = \sum_{m\in\Z^d}(1+ |m|^2)^{-\frac12} \Theta_{e_m,e_m} \otimes 1_{\bigwedge^*\R^d}, 
$$
which is a norm-convergent sum of finite-rank operators and so is compact.
\end{proof}

We call $\lambda^{(d)}$ the fundamental $K$-cycle of the $\Z^d$-action 
because of its similarity to Kasparov's fundamental class~\cite{KasparovNovikov}.

\subsection{The extension $K$-cycle} \label{subsec:extension_module} 
Under mild assumptions on the twist $\theta$ (see~\cite{KR06}), we can unwind the 
crossed product $A= B\rtimes_{\alpha,\theta}\Z^d$ such that, for 
$\alpha = (\alpha^\|,\alpha_d)$ and $\alpha^\|$ the restricted action of $\Z^{d-1}$,
\begin{equation}\label{eq-crossed-iterated}
A =\left(B\rtimes_{\alpha^\|,\theta}\Z^{d-1}\right)\rtimes_{\alpha_d}\Z = C\rtimes_{\alpha_d}\Z
\end{equation}
where $C = B\rtimes_{\alpha^\|,\theta}\Z^{d-1}$. We link $C$ and 
$C\rtimes_{\alpha_d} \Z$ by the Toeplitz extension, which we briefly recall.

Very similar to the construction of the 
crossed product $C\rtimes_{\alpha_d}\Z$, we can consider
$C_{\alpha_d}\N$ the algebra given by finite sums $\sum_{k\in\N} \tilde{S}_d^kc_k  + (\tilde{S}_d^*)^k c_k' $, 
where $c_k,c_k'\in C$ and  $\tilde{S}_d$ is the 
operator such that 
\begin{align*}
 &\tilde{S}_d b = \alpha_d(b) \tilde{S}_d,  &&\tilde{S}_d^* b = \alpha^{-1}_d(b) \tilde{S}_d,  
    &&\tilde{S}_d^*\tilde{S}_d = 1,  &&\tilde{S}_d\tilde{S}_d^* = 1-p
\end{align*}
with $p=p^*=p^2$ a projection. 
Thus $\tilde{S}_d$ is  no longer unitary but an isometry. There is a 
unique $*$-algebra morphism $q:C_{\alpha_d}\N\to C_{\alpha_d}\Z$ 
determined by $ q(\hS_d) = S_d$ and is the identity on $C$. 
Its kernel is the ideal generated by $p$ 
which can easily be seen to be isomorphic to $F\otimes C$ 
where $F$ is the algebra of the finite rank operators. 
The exact sequence
$$
0\to  F\otimes C \to C_{\alpha_d}\N \stackrel{q}{\to} C_{\alpha_d}\Z \to 0
$$
is the algebraic version of the Toeplitz extension, the $C^*$-version is 
obtained by taking the universal $C^*$-closures. The $C^*$-closure 
of $C_{\alpha_d}\N$, denoted by $\calT(\alpha_d)$, is the Toeplitz 
algebra of the $\Z$-action ${\alpha_d}$ and the closure of $F\otimes C$ is 
$\calK\otimes C$, with $\calK$ the algebra 
of complex operators on a separable (real or complex) Hilbert space.
The short exact sequence
\begin{equation} \label{eq:general_Toeplitz_extension}
 0 \to \calK \otimes C \to \calT(\alpha_d) 
  \to C\rtimes_{\alpha_d}\Z \to 0
\end{equation}
gives rise to a class 
$[\mathrm{ext}]$ in the extension group $\mathrm{Ext}^{-1}(A, C)$, 
which by Proposition~\ref{prop:Ext_is_KK^1} is the same as the group 
$KKO^1(A, C)$ 
(or in the complex case $KK^1(A, C)$).
 
The extension class $[\mathrm{ext}]$ serves to 
compute boundary maps in $K$-theory and $K$-homology, namely by taking Kasparov products with it. 
In order to make these maps computable, we construct an unbounded 
representative of $[\mathrm{ext}]$.

\begin{prop}\label{prop:ext_class_is_ext_class}
Let $C$ be a separable and unital $C^*$-algebra and $A = C\rtimes_{\alpha_d} \Z$.
The extension class of the Toeplitz extension of 
Equation \eqref{eq:general_Toeplitz_extension}
is represented by the fundamental $K$-cycle of the $\Z$-action,
\begin{equation} \label{eq:ext_trip} 
 \left( C_{\alpha_d}\Z \hat\otimes C\ell_{0,1},\, \ell^2(\Z,C)_{C} \otimes 
  \bigwedge\nolimits^{\!\ast}\R,\, D=X_1\otimes \gamma^1 \,,
  \gamma_{\bigwedge^* \R}  \right).
\end{equation}
There is an analogous result for complex algebras.
\end{prop}
It serves the clarity of arguments further down to denote 
$D = N\otimes\gamma_{\mathrm{ext}}$ for the extension module.
\begin{proof}
The Kasparov module from Equation \eqref{eq:ext_trip} 
defines the same $KK$-class as the bounded and ungraded cycle 
$$
\left(C \rtimes_{\alpha_d}\Z, \,\ell^2(\Z,C)_{C}, \,2P-1\right)
$$ 
where $P = \chi_{[0,\infty)}(N)$.
From~\cite[Proposition 3.14]{RRS15} we know that this class is the extension class of the 
short exact sequence
\begin{align*} 
  0 \to \End_{C}^0[P(\ell^2(\Z, C))] \to 
    &C^*(P C_{\alpha_d}\Z P,\,\End_{C}^0[P(\ell^2(\Z, C))])  \to Q \to 0,
\end{align*}
where $C^*(P C_{\alpha_d}\Z P,\,\End_{C}^0[P(\ell^2(\Z, C))])$ 
is a closed subalgebra of the adjointable operators, 
$\End_{C}[P(\ell^2(\Z, C))]$, and 
$Q$ the quotient. Now $P$ is the projection onto 
$\ell^2(\N, C)$ and hence $P C_{\alpha_d}\Z P = C_{\alpha_d}\N$. 
Moreover, using~\cite[{\S}13]{Blackadar}
$$
 \End_{C}^0[P(\ell^2(\Z, C))] \cong \calK(\ell^2(\N)) \otimes C.
$$ 
Hence $C^*(P C_{\alpha_d}\Z P,\,\End_{C}^0[P(\ell^2(\Z, C))])$ 
is the completion of $C_{\alpha_d}\N$ and $Q$ is the completion of 
$C_{\alpha_d} \Z$ as required. 
\end{proof}

\subsection{The product with the extension $K$-cycle} 
\label{subsec:edge_module_and_product}
Suppose $d\geq 2$, then we can construct the fundamental $K$-cycles 
$\lambda^{(d-1)}$ and $\lambda^{(d)}$ representing classes in 
$KKO^{d-1}(B\rtimes_{\alpha^{\|},\theta}\Z^{d-1},B)$ 
and $KKO^{d}(B\rtimes_{\alpha,\theta}\Z^d,B)$ respectively. 
The extension class representing Equation \eqref{eq:general_Toeplitz_extension} 
gives a well-defined map 
$$
KKO^{1}(B\rtimes_{\alpha,\theta}\Z^d,C)\times KKO^{d-1}(C,B) \to KKO^d(B\rtimes_{\alpha,\theta}\Z^d,B)
$$
(or complex).
Our central result is that this Kasparov product of the 
extension $K$-cycle with the edge $K$-cycle $\lambda^{(d-1)}$ gives, 
up to a possible sign, $\lambda^{(d)}$,
the fundamental $K$-cycle of our original crossed product 
$A =B\rtimes_{\alpha,\theta}\Z^{d}$.
\newcommand{\CO}{C}
\begin{thm} \label{thm:KKO_factorisation}
Let $B$ be a separable and unital real or complex $C^*$-algebra with 
fundamental $K$-cycles $\lambda^{(d)}$ and $\lambda^{(d-1)}$ 
for (possibly twisted) $\Z^d$ and $\Z^{d-1}$-actions. 
Then the unbounded Kasparov product of the extension 
Kasparov module from Proposition \ref{prop:ext_class_is_ext_class} 
with $\lambda^{(d-1)}$ gives, up to a cyclic permutation of the 
Clifford generators, the fundamental $K$-cycle $\lambda^{(d)}$.
On the level of $KK$-classes this means 
$$
  [\ext] \hat\otimes_{C} [\lambda^{(d-1)}] = (-1)^{d-1} [\lambda^{(d)}],
$$
where $-[x]$ denotes the inverse of the $KK$-class.
\end{thm}
\begin{proof}
We will focus on the real setting though note that the case of complex algebras and 
modules follows the same argument. We denote 
by $C = B\rtimes_{\alpha^\|,\theta}\Z^{d-1}$ and 
$A = B\rtimes_{\alpha,\theta}\Z^d \cong C \rtimes_{\alpha_d} \Z$ 
with dense subalgebras $\calC=B_{\alpha^\|,\theta}\Z^{d-1}$ and 
$\calA=B_{\alpha,\theta}\Z^d$.
We are taking the product of an $A\hat\otimes C\ell_{0,1}$-$C$ Kasparov module 
with a $C \hat\otimes C\ell_{0,d-1}$-$B$ Kasparov module. Our first step is an 
external product of the $A \hat\otimes C\ell_{0,1}$-$C$ Kasparov module 
with the identity class in $KKO(C\ell_{0,d-1},C\ell_{0,d-1})$.
This class can be represented by the Kasparov module
$$ 
 \left( C\ell_{0,d-1}, \left(C\ell_{0,d-1}\right)_{C\ell_{0,d-1}},0,\gamma_{C\ell_{0,d-1}}\right) 
$$
with right and left actions given by right and left Clifford multiplication. 
The external product results in the 
$A\hat\otimes C\ell_{0,d}$-$C\hat\otimes C\ell_{0,d-1}$ Kasparov module represented by 
\begin{align*}
 &\bigg(\calA \hat\otimes C\ell_{0,d},\left(\ell^2(\Z,C)\otimes \bigwedge\nolimits^{\!*}\R\, \hat\otimes 
  C\ell_{0,d-1}\right)_{C \hat\otimes C\ell_{0,d-1}}, 
    N \otimes \gamma_{\mathrm{ext}} \hat\otimes 1, \,
  \gamma_{\bigwedge^*\R} \hat\otimes \gamma_{C\ell_{0,d-1}} \bigg).
\end{align*}
We now take the internal product of this module with the $K$-cycle for 
twisted $\Z^{d-1}$-actions,  $[\lambda^{(d-1)}] \in KKO(C\hat\otimes C\ell_{0,d-1},B)$ 
and represented by the unbounded Kasparov module
$$
 \bigg( \calC\hat\otimes C\ell_{0,d-1},\, \ell^2(\Z^{d-1},B) 
   \otimes \bigwedge\nolimits^{\!*}\R^{d-1},\, \sum_{j=1}^{d-1} X_j \otimes \gamma^j, \,
     \gamma_{\bigwedge^*\R^{d-1}} \bigg) .
$$
We start with the balanced tensor product of $C^*$-modules, where
\begin{align*}
 &\left( \ell^2(\Z,C)\otimes_\R \bigwedge\nolimits^{\!*}\R \,\hat\otimes_\R \, C\ell_{0,d-1}\right) 
    \hat\otimes_{C\hat\otimes C\ell_{0,d-1}} 
      \left(\ell^2(\Z^{d-1},B) \otimes_\R \bigwedge\nolimits^{\!*}\R^{d-1} \right) \\
 &\hspace{0.3cm}  \cong \left( \ell^2(\Z, C)\otimes_{C} \ell^2(\Z^{d-1},B)\right) 
   \otimes_\R \bigwedge\nolimits^{\!*}\R \,\hat\otimes_\R 
      \left(C\ell_{0,d-1}\cdot \bigwedge\nolimits^{\!*}\R^{d-1} \right) \\
 &\hspace{0.3cm} \cong \left( \ell^2(\Z, C)\otimes_{C} \ell^2(\Z^{d-1},B)\right)
 \otimes_\R \bigwedge\nolimits^{\!*}\R\, \hat\otimes_\R \bigwedge\nolimits^{\!*}\R^{d-1}
\end{align*}
as the action of $C\ell_{0,d-1}$ on $\bigwedge^*\R^{d-1}$ by left-multiplication is nondegenerate.

Next we define $1\otimes_\nabla X_j$ on the dense 
submodule $\ell^2(\Z, \calC)\otimes_{\calC} \ell^2(\Z^{d-1},B)$ 
for all $j\in\{1,\ldots,d-1\}$, where
\begin{align} \label{eq:twisted_dirac}
   &(1\otimes_\nabla X_j)\left(\delta_k \otimes c \otimes \delta_n \otimes b \right) 
   =  \delta_k \otimes c \otimes X_j \delta_n \otimes b
      + \delta_k\otimes 1\otimes [X_j,c]\delta_n \otimes b  
\end{align}
with $\{\delta_k\}_{k\in\Z}$ an orthonormal basis of $\ell^2(\Z)$, $c \in \calC$, 
$\{\delta_n\}_{n\in\Z^{d-1}}$ an orthonormal basis for $\ell^2(\Z^{d-1})$ and $b\in B$. 
The operator is well-defined as $[X_j,c]\in\calC$ for $c\in\calC$. Then
\begin{align} \label{eq:prod_module_candidate}
  &\bigg( \calA \hat\otimes C\ell_{0,1}\hat\otimes C\ell_{0,d-1},\, 
     \left( \ell^2(\Z, C)\otimes_{C} \ell^2(\Z^{d-1},B)\right) 
     \otimes \bigwedge\nolimits^{\!*}\R\, \hat\otimes_\R \bigwedge\nolimits^{\!*}\R^{d-1}, \\
  &\hspace{7cm}  N\otimes 1 \otimes \gamma_{\mathrm{ext}}\hat\otimes 1 
     + \sum_{j=1}^{d-1}(1\otimes_\nabla X_j)\otimes 1\hat\otimes \gamma^j \bigg) \nonumber
\end{align}
is a candidate for the unbounded Kasparov module representing the product, 
where the Clifford actions take the form
\begin{align*}
   \rho_{\mathrm{ext}}\hat\otimes 1(\omega_1\hat\otimes\omega_2) &= (e_1\wedge \omega_1 - \iota(e_1)\omega_1)\hat\otimes \omega_2 \\
   1\hat\otimes \rho^j(\omega_1\hat\otimes\omega_2) &= (-1)^{|\omega_1|} \omega_1\hat\otimes(e_j\wedge \omega_2 - \iota(e_j)\omega_2),
\end{align*}
for $j\in\{1,\ldots,d-1\}$ with $|\omega_1|$ is the degree of the form $\omega_1$. 
Similar equations hold for $\gamma_\mathrm{ext}\hat\otimes 1$ and $1\hat\otimes\gamma^j$. 
Arguments very 
similar to the proof of Proposition \ref{prop:general-toeplitz} show that Equation 
\eqref{eq:prod_module_candidate} is a real or complex Kasparov module depending on what 
setting we are in.
We now briefly check 
Kucerovsky's criterion (Theorem \ref{thm:Kucerovsky_criterion}). The connection criterion 
holds precisely because we have used a connection $\nabla$ to construct $1\otimes_\nabla X_j$. 
The domain condition is a simple check and the positivity condition 
is explicitly checkable as the operators 
of interest act as number operators. Therefore the unbounded Kasparov module of Equation 
\eqref{eq:prod_module_candidate}
is an unbounded representative of the product $[\text{ext}]\hat\otimes_{C}[\lambda^{(d-1)}]$ 
at the level of $KK$-classes.

Our next task is to relate the Kasparov module of Equation \eqref{eq:prod_module_candidate} to 
$\lambda^{(d)}$. 
We first identify  $\bigwedge^*\R\, \hat\otimes_\R \bigwedge^* \R^{d-1} \cong \bigwedge^*\R^d$ 
and use the graded isomorphism $C\ell_{p,q}\hat\otimes C\ell_{r,s} \cong C\ell_{p+r,q+s}$ 
from~\cite[{\S}2.16]{Kasparov80}
on the left and right Clifford generators by the mapping
\begin{align*}
   &\rho_{\mathrm{ext}}\hat\otimes 1 \mapsto \rho^1,  &&1\hat\otimes \rho^j \mapsto \rho^{j+1}, \\
   &\gamma_{\mathrm{ext}} \hat\otimes 1 \mapsto \gamma^1,  &&1\hat\otimes \gamma^j \mapsto \gamma^{j+1}.
\end{align*} 
Therefore applying this isomorphism gives the unbounded Kasparov module
\begin{align*} 
 &\bigg( \calA \hat\otimes C\ell_{0,d}, \,(\ell^2(\Z, C)\otimes_{C} \ell^2(\Z^{d-1},B)) 
 \otimes \bigwedge\nolimits^{\!*}\R^d, \,
    N\otimes 1 \otimes \gamma^1 + 
  \sum_{j=1}^{d-1} (1\otimes_\nabla X_{j}) \otimes \gamma^{j+1} \bigg)
\end{align*}
with $C\ell_{0,d}$-action generated by $\rho^j(\omega) = e_j\wedge\omega -\iota(e_j)\omega$ and 
$C\ell_{d,0}$-action generated by $\gamma^j(\omega) = e_j\wedge\omega +\iota(e_j)\omega$ for 
$\omega\in\bigwedge^*\R^d$ and $\{e_j\}_{j=1}^d$ the standard basis of $\R^d$.

Next we define the map $\ell^2(\Z, C)\otimes_{C} \ell^2(\Z^{d-1},B) 
\to \ell^2(\Z^d, B)$ on generators as
$$
  \delta_k \otimes c \otimes \delta_n \otimes b \mapsto \alpha_d^k(c)\cdot \delta_{(n,k)}\otimes b
$$
for $k\in\Z$, $n\in\Z^{d-1}$ and $\alpha_d$ the automorphism on $C$ such that 
$A = C\rtimes_{\alpha_d}\Z$.
One easily checks that this map is compatible with the right-action of $B$ and $B$-valued 
inner product and, thus, is a unitary map on $C^*$-modules. We check compatibility with the 
left action by $A \cong C^*(C, S_d)$, where for $c_1,c_2\in C$,
\begin{align*}
  c_1(\delta_k \otimes c_2 \otimes \delta_n \otimes b) &= \delta_k \otimes \alpha_d^{-k}(c_1) c_2
    \otimes \delta_n \otimes b \\
   &\mapsto \alpha_d^k\!\left(\alpha_d^{-k}(c_1)c_2\right) \cdot \delta_{(n,k)}\otimes b \\
   &= c_1 \left(\alpha_d^k(c_2)\cdot \delta_{(n,k)}\otimes b \right).
\end{align*}
Furthermore,
\begin{align*}
  S_d(\delta_k \otimes c \otimes \delta_n \otimes b) &= \delta_{k+1} \otimes c \otimes \delta_n \otimes b \\
    &\mapsto \alpha_d^{k+1}(c)\cdot \delta_{(n,k+1)}\otimes b \\
    &= S_d \alpha_d^k(c) S_d^* \cdot \delta_{(n,k+1)} \otimes b \\
    &= S_d\left( \alpha_d^k(c) \cdot \delta_{(n,k)} \otimes b \right)
\end{align*}
and so the generators of the left-action are compatible with the isomorphism. 
It is then straightforward to check that 
the compatibility of the left-action extends to the $C^*$-completion. Next we compute that 
\begin{align*}
  N(\delta_k \otimes c \otimes \delta_n \otimes b) &= k \delta_k \otimes c \otimes \delta_n \otimes b \\
    &\mapsto k \alpha_d^k(c)\cdot \delta_{(n,k)}\otimes b = X_d\left( \alpha_d^k(c)\cdot \delta_{(n,k)}\otimes b\right)
\end{align*}
as $\alpha_d^k(c)$ commutes with $X_d$. Lastly we use Equation \eqref{eq:twisted_dirac} 
to check that
\begin{align*}
  (1\otimes_\nabla X_j)(\delta_k \otimes c \otimes \delta_n \otimes b) 
   &=  \delta_k \otimes (c + [X_j,c]) \otimes \delta_n \otimes b  \\
   &\mapsto  \alpha_d^k(c + [X_j,c])\cdot \delta_{(n,k)} \otimes b \\
   &= \alpha_d^k(X_j c)\cdot \delta_{(n,k)} \otimes b \\
   &= X_j\left( \alpha_d^k(c)\cdot \delta_{(n,k)}\otimes b \right)
\end{align*}
as $[X_j,c]\in\calC$ for $c\in \calC$ 
and $\alpha_d^k(X_j) = X_j$ for $j\in\{1,\ldots,d-1\}$.
We note that the identification $1\otimes_\nabla X_j \mapsto X_j$ requires us to use the 
smooth subalgebra $\calC \subset C$ in order for $[X_j,c]$ 
to be a well-defined element of $C$. This is typical of the unbounded product and why we need to have the equality 
$\ol{\calC \cdot \ell^2(\Z^{d-1},B)} = \ell^2(\Z^{d-1},B)$.

To summarise, the unbounded Kasparov module representing the product 
is unitarily equivalent to 
\begin{equation} \label{eq:prod_module_simplified}
  \bigg( \calA\hat\otimes C\ell_{0,d},\, \ell^2(\Z^d, B)_B
 \otimes \bigwedge\nolimits^{\!*}\R^d,\, X_d\otimes \gamma^1 + \sum_{j=1}^{d-1} X_{j} \otimes \gamma^{j+1}, \, \gamma_{\bigwedge^*\R^d} \bigg)
\end{equation}
with left and right Clifford actions as previously. This is almost the same as the 
fundamental $K$-cycle $\lambda^{(d)}$, the only difference being the
labelling 
of the Clifford basis. The map $\xi(\gamma^j)= \gamma^{\sigma(j)}$ and 
$\xi(\rho^j) = \rho^{\sigma(j)}$ for $\sigma(j) = (j-1)\mathrm{mod}\,d$ is potentially 
an orientation reversing map on Clifford algebras. Taking the canonical orientation 
 $\omega_{C\ell_{0,d}}=\rho^1\cdots\rho^d$ on $C\ell_{0,d}$, 
$$
  \xi(\omega_{C\ell_{0,d}}) = \rho^d \rho^1\cdots\rho^{d-1} = (-1)^{d-1} \rho^1\cdots \rho^d 
   = (-1)^{d-1} \omega_{C\ell_{0,d}},
$$
and similarly for the $\gamma^j$ and $C\ell_{d,0}$. Using~\cite[{\S}5, Theorem 3]{Kasparov80}, such 
a map on Clifford algebras will send the $KK$-class of the Kasparov module of Equation 
\eqref{eq:prod_module_simplified} to its inverse if $d$ is even or leaves 
the $KK$-class invariant if $d$ is odd. 
Therefore at the level 
of $KK$-classes, $[\mathrm{ext}]\hat\otimes_{C} [\lambda^{(d-1)}] = (-1)^{d-1}[\lambda^{(d)}]$ 
as required.
\end{proof}

\begin{remark}
It would be preferable to have an explicit unitary equivalence implementing
the coordinate permutation of Clifford indices as existing proofs demonstrating the 
signed equality arising from a permutation require (explicit) homotopies. The difficulty
of defining reasonable equivalence relations on unbounded Kasparov modules stronger
than `unitary equivalence modulo bounded perturbation' means that such homotopies
compromise the computability and any interpretation stronger than mere (signed) equality of
$KK$-classes.
\end{remark}

\begin{cor}
The pairing of a $K$-theory class $[z]\in KO_j(B\rtimes_{\alpha,\theta}\Z^d)$ 
(or complex) with $\lambda^{(d)}$ is, up to a sign, the same as the 
pairing of $\partial[z] \in KO_{j-1}(B\rtimes_{\alpha^\|,\theta}\Z^{d-1})$ 
with $\lambda^{(d-1)}$.
\end{cor}
\begin{proof}
Using Theorem \ref{thm:KKO_factorisation} and associativity of the 
Kasparov product,
\begin{align*}
  [z]\hat\otimes_A [\lambda^{(d)}] &= (-1)^{d-1}[z]\hat\otimes_A \big([\mathrm{ext}] 
    \hat\otimes_{C} [\lambda^{(d-1)}] \big) \\
    &= (-1)^{d-1} \big( [z]\hat\otimes_A [\mathrm{ext}] \big) 
      \hat\otimes_{C}  [\lambda^{(d-1)}] \\
    &= (-1)^{d-1} \partial[z] \hat\otimes_{C}  [\lambda^{(d-1)}]
\end{align*}
as the product with $[\mathrm{ext}]$ implements the boundary map in 
$K$-theory.
\end{proof}
Similarly if $d=1$ and $A=B\rtimes_\beta\Z$, then $[\lambda^{(1)}] = [\mathrm{ext}]$ and 
$[z]\hat\otimes_A [\lambda^{(1)}] = \partial[z]$ which is in $KO_{j-1}(B\otimes\calK)$, 
the $K$-theory of the `$0$-dimensional boundary'. 
We will return to such pairings in Section \ref{sec:pairings_and_bulkedge}.

\section{The disordered bulk-edge correspondence in $K$-homology} \label{sec:disorder_bulkedge}
We now apply our results about fundamental $K$-cycles of crossed products and the 
Pimsner--Voiculescu exact sequence to the case of $C^*$-algebras 
modelling disordered or aperiodic topological phases.

\subsection{The algebras and the extension for topological insulators}
\label{sec:disorder_algebra_and_ext}

\subsubsection{The bulk algebra}
The bulk algebra is the observable algebra of the solid seen as infinite and without boundary. 
We denote it by $A_b$.
It is also referred to as the noncommutative Brillouin zone and was  developed 
in~\cite{BelGapLabel, Bellissard94}. 
We work in the tight binding approximation where the configuration space is $\Z^d$. 
In this case the bulk algebra is $A_b=C(\Omega,M_N)\rtimes_{\alpha,\theta} \Z^d$,
the twisted crossed product algebra of
matrix-valued continuous functions over a compact space $\Omega$ by the group $\Z^d$.

The space $\Omega$ may be thought of as a space of disorder configurations. 
It is equipped with a topological structure and a continuous action $\alpha$ 
of $\Z^d$ by shift of the configuration. One usually assumes that $\Omega$
possesses a probability measure that is invariant and ergodic under 
the $\Z^d$-action to obtain expectation values. 
The space $\Omega$ could be taken to be a point, or contractible to a point. 
The latter point of view is taken in \cite{PSBbook} but this excludes the interesting 
possibility of describing for instance quasicrystals.
The algebra of $N\times N$ matrices, $M_N$, is used to incorporate 
internal degrees of freedom like spin. Whereas these internal 
components are important for the implementation of certain symmetries, 
for instance odd time reversal, they do not interfere with the main 
topological constructions which will follow and, in 
order not to overburden the notation, we will suppress them.

An external magnetic field can be incorporated by a two 
cocycle $\theta:\Z^d\times\Z^d \to \calU(C(\Omega))$ to twist 
the action, but note that the existence of anti-linear symmetries 
puts restrictions on the form of such a magnetic field. Indeed, 
for systems with symmetries which are implemented by anti-linear 
operators we are bound to look at a real subalgebra of the above algebra. 
In this case we consider therefore $C(\Omega)$ as the 
algebra of real valued continuous functions and exclude the 
possibility of complex-valued twisting cocycles. The presence of a boundary
will impose further constraints on any possible cocycle.

Of importance for our application to physics is the family of representations  
$\{\pi_\omega\}_{\omega\in\Omega}$ which are induced by the 
evaluation representations $\mathrm{ev}_\omega:C(\Omega)\to\R$ (or $\C$). 
These $\pi_\omega$ are thus representations of $C(\Omega)\rtimes_{\alpha,\theta}\Z^d$
on $\ell^2(\Z^d)$. They look as follows 
on $C(\Omega)_{\alpha,\theta}\Z^d$, a dense subalgebra of the 
bulk algebra,
\begin{align*} 
 &\pi_\omega(f)\delta_n = f(\alpha^{n}(\omega))\delta_n, 
 &&\pi_\omega(S_i)\delta_n = \theta(n,e_i)(\omega)\delta_{n+e_i}
\end{align*}
with $f\in C(\Omega)$ and $\{S_i\}_{i=1}^d$ the unitary generators of the 
twisted $\Z^d$-action with $\{e_i\}_{i=1}^d$ the generators of $\Z^d$.
A tight binding Hamiltonian is a self-adjoint element of the algebraic crossed product
\begin{equation*}
  h = \sum_{n\in\Z^d} {S}^n v_{n} .
\end{equation*}
Here the functions $v_n\in C(\Omega)$ are chosen such that the expression is 
self-adjoint and the sum has only finitely many non-zero terms. If we 
take internal degrees of freedom into account then the $v_n$ are matrix valued 
and $h\in M_k(C(\Omega)_{\alpha,\theta}\Z^d)$. 
The representations $\pi_\omega$ give rise to a family of Hamiltionians 
$H_\omega = \pi_\omega(h)$ which are unitarily equivalent for points 
$\omega$ in the same orbit. 

\subsubsection{The Toeplitz extension of the bulk algebra}

We require that the twisting cocycle $\theta$ can be arranged in such a way that 
$S_iS_d = S_dS_i$ for all $i\in\{1,\ldots,d\}$, something that often reduces 
to a choice of magnetic gauge. 
This assumption allows us to decompose the twisted $\Z^d$-action,
$\alpha = (\alpha^\|,\alpha_d)$ with $\alpha^\|$ a twisted $\Z^{d-1}$-action 
and $\alpha_d$ an untwisted $\Z$-action, i.e.,
$$
  C(\Omega)\rtimes_{\alpha,\theta}\Z^d 
  \cong \left(C(\Omega)\rtimes_{\alpha^\|,\theta} \Z^{d-1}\right) \rtimes_{\alpha_d} \Z.
$$
The bulk algebra can then be seen as a quotient algebra of the 
Toeplitz algebra $\calT(\alpha_d)$, which is interpreted as the observables 
on a system with boundary.
Indeed, the evaluation representation 
$\mathrm{ev}_\omega:C(\Omega)\to \R$ (or $\C$) induces a representation of 
$(C(\Omega)_{\alpha^\|,\theta}\Z^{d-1})_{\alpha_d}\N$ on $\ell^2(\N)\otimes\ell^2(\Z^{d-1})$. 
The tensor product Hilbert space is dense in $\ell^2(\Z^{d-1}\times\N)$ and 
$\Z^{d-1}\times\N$ is the configuration space for the 
system on a half-space. That is, the insulator seen as 
infinite with a boundary, the so-called edge, which is at 
$\Z^{d-1}\times\{0\}$. The representation of 
$(C(\Omega)_{\alpha^\|,\theta}\Z^{d-1})_{\alpha_d}\N$ looks far away from the edge like 
that of $(C(\Omega)_{\alpha^\|,\theta}\Z^{d-1})_{\alpha_d}\Z=C(\Omega)_{\alpha}\Z^d$, but at the edge the 
translation $\tilde{S}_d$ is truncated. 
The algebra $(C(\Omega)_{\alpha^\|,\theta}\Z^{d-1})_{\alpha_d}\N$, or rather its closure $\calT(\alpha_d)$ 
can therefore be seen as the observable algebra of the insulator with boundary.

\subsubsection{The edge algebra}
The projection $p = 1-\tilde{S}_d \tilde{S}_d^*$
generates a ideal in $\calT(\alpha_d)$.
In the evaluation representation this projection 
becomes the projection onto the subspace 
$\ell^2(\Z^{d-1}\times\{0\})$, and 
hence is naturally associated to the edge.  
The algebraic (non-closed) ideal is given by 
$F\otimes C(\Omega)_{\alpha^\|,\theta}\Z^{d-1}$ with 
$F$ the finite-rank operators perpendicular to the edge.
 Its $C^*$-closure
$\calK\otimes C(\Omega)\rtimes_{\alpha^\|,\theta}\Z^{d-1}$ 
is represented as operators which are 
localised near the edge and so is called the edge algebra.
We will use the notation $A_e =  C(\Omega)\rtimes_{\alpha^\|,\theta}\Z^{d-1}$ so that the edge 
algebra is the stabilisation of $A_e$. 

To summarize, the extension which gives rise to the bulk boundary correspondence for topological insulators is
$$
0\to \calK\otimes C(\Omega)\rtimes_{\alpha^\|,\theta}\Z^{d-1} \to\calT(\alpha_d)\to 
  C(\Omega)\rtimes_{\alpha,\theta}\Z^{d}\to 0.
$$

\subsection{The fundamental $K$-cycles for the bulk and the edge}
We apply the construction of Proposition \ref{prop:general-toeplitz} 
to $A_b=C(\Omega)\rtimes_{\alpha,\theta}\Z^{d}$ 
and to $A_e=C(\Omega)\rtimes_{\alpha^\|,\theta}\Z^{d-1}$ with their dense subalgebras
$\calA_b = C(\Omega)_{\alpha,\theta}\Z^d$ and 
$\calA_e = C(\Omega)_{\alpha^\|,\theta}\Z^{d-1}$. As a result we obtain the fundamental $K$-cycle for the bulk
\begin{equation} \label{eq:fund-cl-bulk} 
\lambda_b =  \left(\calA_b \hat\otimes C\ell_{0,d},\, \ell^2(\Z^{d},C(\Omega))_{C(\Omega)} \otimes 
  \bigwedge\nolimits^{\!\ast}\R^{d},\, D_b, \,
  \gamma_{\bigwedge^* \R^d}  \right), \quad D_b = \sum_{j=1}^d X_j\otimes \gamma^j,
\end{equation}
and the fundamental $K$-cycle of the edge
\begin{equation*}  
\lambda_e = \left(\calA_e \hat\otimes C\ell_{0,d-1}, \, \ell^2(\Z^{d-1},C(\Omega))_{C(\Omega)} \! \otimes 
  \bigwedge\nolimits^{\!\ast}\R^{d-1}, \, D_e, \, 
  \gamma_{\bigwedge^* \R^{d-1}}  \right), \quad D_e =\sum_{j=1}^{d-1}X_j\otimes \gamma^j.
\end{equation*}
Hence we obtain unbounded representatives of 
elements $[\lambda_b] \in KKO^d(A_b,C(\Omega))$ and 
$[\lambda_e] \in KKO^{d-1}(A_e,C(\Omega))$ (in the real case, otherwise it's $KK$).

We apply Proposition \ref{prop:ext_class_is_ext_class} to $A_b = A_e\rtimes_{\alpha_d}\Z$ 
to obtain the extension cycle
\begin{equation*} 
 \ext_{be} = \left( \calA_b \hat\otimes C\ell_{0,1},\, \ell^2(\Z,A_e)_{A_e} \otimes 
  \bigwedge\nolimits^{\!\ast}\R,\, D=X_1\otimes \gamma^1 \,,
  \gamma_{\bigwedge^* \R}  \right).
\end{equation*} 
Theorem~\ref{thm:KKO_factorisation} then yields the factorisation of the fundamental $K$-cycles 
for topological insulators,
\begin{equation*}
 [\ext_{be}] \hat\otimes [\lambda_e] = (-1)^{d-1} [\lambda_b].
\end{equation*}

\subsection{From Kasparov module to spectral triple}
Recall the map $\mathrm{ev}_\omega:C(\Omega)\to \mathbb{F}$ (for 
$\mathbb{F} = \R$ or $\C$), which 
gives a family of representations $\{\pi_\omega\}_{\omega\in\Omega}$ 
of the crossed product $C(\Omega)\rtimes_{\alpha,\theta}\Z^d$ on 
$\ell^2(\Z^d)$ subject to the covariance condition, 
$S^n \pi_\omega(a) (S^*)^n = \pi_{\alpha^n(\omega)}(a)$. 
The representations $\pi_\omega$ were used in~\cite{GSB15,BCR15} 
to define a real or complex spectral triple
\begin{equation} \label{eq:old_bulk_spec_trip}
 \lambda_b(\omega) = \bigg( \calA\hat\otimes C\ell_{0,d}, \, {}_{\pi_\omega}\ell^2(\Z^d) \otimes 
   \bigwedge\nolimits^{\!*}\R^d, \, \sum_{j=1}^d X_j \otimes \gamma^j, \, 
   \gamma_{\bigwedge^* \R^d} \bigg).
\end{equation}
The spectral triple gives a $K$-homology class for the crossed-product 
algebra, which by the covariance relation is independent under a fixed 
$\alpha$-orbit of $\Omega$.

Let us link the bulk spectral triple of Equation \eqref{eq:old_bulk_spec_trip} 
to the fundamental $K$-cycle of the bulk algebra from Equation \eqref{eq:fund-cl-bulk}.
We take the trivially graded Kasparov module
$$
 \mathrm{ev}_\omega =  \left( C(\Omega), \, {}_{\mathrm{ev}_\omega}\mathbb{F},\,0 \right)
$$
with $\mathrm{ev}_\omega$ the evaluation representation of $C(\Omega)$. 
The Kasparov module $\mathrm{ev}_\omega$ represents a class in $KKO(C(\Omega),\R)$ 
(or $KK(C(\Omega),\C)$) that can be paired with our fundamental $K$-cycle.

\begin{prop}
The spectral triple $\lambda_b(\omega)$ is unitarily equivalent to the internal product of the 
unbounded bulk Kasparov module $\lambda_b$ with $\mathrm{ev}_\omega$.
\end{prop}
\begin{proof}
Because the evaluation Kasparov module is very simple, we can easily compute 
the internal product
\begin{align*}
   &\bigg(\calA \hat\otimes C\ell_{0,d},\, \ell^2(\Z^{d},C(\Omega)) \otimes 
     \bigwedge\nolimits^{\!\ast}\R^{d},\, \sum_{j=1}^d X_j\otimes \gamma^j, \,
     \gamma_{\bigwedge^* \R^d}  \bigg)  \hat\otimes_{C(\Omega)}
     \left( C(\Omega), \, {}_{\mathrm{ev}_\omega}\mathbb{F},\,0 \right) \\
  &\hspace{1cm}\cong \bigg( \calA \hat\otimes C\ell_{0,d},\, 
      \calH  \otimes 
      \bigwedge\nolimits^{\!\ast}\R^{d},\, \sum_{j=1}^d X_j\otimes1\otimes \gamma^j, \,
     \gamma_{\bigwedge^* \R^d}  \bigg)
\end{align*}
with $\calH = \ell^2(\Z^d,C(\Omega)) \otimes_{\mathrm{ev}_\omega} \mathbb{F}$.
We identify $\ell^2(\Z^d,C(\Omega)) \otimes_{\mathrm{ev}_\omega} \mathbb{F}$ 
with $\ell^2(\Z^d)$ 
under which $X_j\otimes 1\mapsto X_j$ and the left-action of $\calA$ takes the form
$$  
  S^m g\cdot \delta_n = \theta(n,m) g(\alpha^n(\omega))\delta_{n+m} = \pi_\omega(S^m g)\delta_n.
$$
Hence we recover $\lambda_b(\omega)$.
\end{proof}

Similarily we obtain a spectral triple for the edge algebra
$$
 \lambda_e(\omega) \!=\! \bigg( \calA_e \hat\otimes C\ell_{0,d-1}, {}_{\pi_\omega}\ell^2(\Z^{d-1}) 
   \otimes \bigwedge\nolimits^{\!*}\R^{d-1},  \sum_{j=1}^{d-1} X_j \otimes \gamma^j,  
   \gamma_{\bigwedge^*\R^{d-1}} \bigg)
$$  
with $[\lambda_e(\omega)] = [\lambda_e]\hat\otimes_{C(\Omega)} [\mathrm{ev}_\omega]$.
Therefore by Theorem \ref{thm:KKO_factorisation} we obtain a 
factorisation of the bulk spectral triple 
\begin{equation*}
 [\ext_{be}] \hat\otimes [\lambda_e(\omega)] = (-1)^{d-1} [\lambda_b(\omega)].
\end{equation*}


\section{$K$-theory, pairings and computational challenges}
\label{sec:K_theory_and_pairings}
In this section we outline how our rather general results about Kasparov 
modules and twisted $\Z^d$-actions can be related to topological phases 
and the bulk-edge correspondence. 

As mentioned in the introduction, the 
measured quantities in insulator models involve the pairing of a $K$-theory 
class with a dual theory, be it cyclic cohomology, $K$-homology or the 
more general Kasparov product. 

Computing these pairings gives rise to 
analytic index formulas, which are in general a Clifford module valued 
index in the sense of Atiyah--Bott--Shapiro~\cite{ABS64}. In the case of 
non-torsion invariants, cyclic formulas may be used to obtain more computable 
expressions for disorder-averaged quantities. The case of torsion indices is 
more complicated and we finish with some brief remarks about computing 
such invariants.

\subsection{Symmetries and $K$-theory}
Recall that our Hamiltonian of interest is a self-adjoint element $h$ in the 
crossed-product algebra $C(\Omega)\rtimes_{\alpha,\theta}\Z^d$. 
A Hamiltonian $h$ with spectral gap at $0$ represents an 
extended topological phase if $h$ is compatible with 
certain symmetry involutions. The symmetries of interest to us 
are time-reversal symmetry (TRS), particle-hole (charge-conjugation) symmetry (PHS) and
chiral (sublattice) symmetry, although we emphasise that 
other symmetries such as spatial involution may be considered. The following result,
due in various forms to numerous authors,
associates a $K$-theory class to a symmetry compatible Hamiltonian.

\begin{prop}[\cite{BCR15, FM13, Kellendonk15, Kubota15b, Thiang14}] 
\label{prop:symmetries_give_K-theory}
Let $h\in C(\Omega)\rtimes_{\alpha,\theta}\Z$ be self-adjoint with 
a spectral gap at $0$. If $h$ is compatible with time-reversal symmetry and/or 
particle-hole symmetry and/or chiral symmetry then we may associate a class 
in $K_j(C(\Omega)\rtimes\Z^d)$ or $KO_j(C(\Omega)\rtimes\Z^d)$, where 
$j$ is determined by the symmetries present. The details are summarised 
in Table \ref{table:K_theory_classification}.
\end{prop}

The specific class associated to a symmetry compatible Hamiltonian 
in Proposition \ref{prop:symmetries_give_K-theory} may arise in 
several (largely equivalent) ways.

One may consider a subgroup $G$ of the $CT$-symmetry 
group $\{1,T,C,CT\}\cong \Z_2\times \Z_2$ with $C$ denoting charge-conjugation (particle-hole), 
$T$ time-reversal and $CT$ sublattice/chiral symmetry. A topological 
phase can then be considered as a Hamiltonian with spectral gap (assumed to 
be at $0$), whose phase $h|h|^{-1}$ acts as a grading for a projective 
unitary/anti-unitary (PUA) representation of $G\subset\{1,T,C,CT\}$ on 
a complex Hilbert space $\calH$. This is the 
perspective first developed in~\cite{FM13} in the commutative setting and then 
extended to noncommutative algebras in~\cite{BCR15,Thiang14}.
The even/odd nature of the time-reversal and charge-conjugation symmetries 
is encoded in the cocycle that arises in the projective representation. 
The class in $K_j(A)$ or $KO_j(A)$ for $A\subset C(\Omega)\rtimes_{\alpha,\theta}\Z^d$ 
from Proposition \ref{prop:symmetries_give_K-theory} and Table \ref{table:K_theory_classification}
is determined by the subgroup $G$ of the $CT$-symmetry group 
and its PUA representation.

\begin{table} 
    \centering
    \begin{tabular}{ p{4.2cm} |  p{3.0cm} }
       {Symmetry type} & {$K$-theory group} \\ \hline
         even TRS & $KO_0(A)$  \\
         even TRS, even PHS & $KO_1(A)$ \\
         even PHS & $KO_2(A)$ \\
         even PHS, odd TRS & $KO_3(A)$ \\
         odd TRS & $KO_4(A)$ \\
         odd TRS, odd PHS & $KO_5(A)$ \\
         odd PHS & $KO_6(A)$ \\
         even TRS, odd PHS & $KO_7(A)$ \\   \hline \hline 
         N/A &  $K_0(A)$ \\
         chiral & $K_1(A)$
    \end{tabular}
    \caption{Symmetry types and the corresponding 
    $K$-theory group of the ungraded algebra 
    $A=C(\Omega)\rtimes_{\alpha,\theta}\Z^d$ (real or complex). 
    See Proposition \ref{prop:symmetries_give_K-theory}.
    \label{table:K_theory_classification}}
\end{table}

One may also consider the gapped Hamiltonian $h$ to be oddly graded 
by the presence of a chiral/sublattice symmetry, in which case the 
phase $h|h|^{-1}$ determines a class in the 
van Daele $K$-theory of the algebra, which can then 
be related to real or complex $K$-theory.\footnote{
If there is no chiral/sublattice symmetry then for $h$ to have odd grading 
we must first consider a larger graded algebra which we can then reduce 
to the trivially graded $A_b$. 
See~\cite{Kellendonk15}.}
Other symmetries can be incorporated 
by considering Real structures on the observable algebra. By considering 
the various types of symmetries and whether they are even and odd, one can 
derive Table \ref{table:K_theory_classification}. 
See~\cite{Kellendonk15,Kubota15a} for further information.

The generality and flexibility of $KK$-theory means that we can use either 
the symmetry class or the van Daele class in 
terms of pairing or the bulk-edge correspondence. The choice of $K$-theory 
class can therefore be determined in order to model the specifics of an
experimental set-up. This is important as the techniques being employed
to define/measure $\Z_2$-invariants of topological phases are still 
in development.

\subsection{Pairings, the Clifford index and the bulk-edge correspondence}
\label{sec:pairings_and_bulkedge}
As briefly explained, a symmetry compatible gapped Hamiltonian gives 
rise to a class in $K_j(A_b)$ or $KO_j(A_b)$ for 
$A_b$ the real or complex ungraded 
crossed product $C(\Omega)\rtimes_{\alpha,\theta}\Z^d$. 
By Theorem \ref{prop:real_Real_k_with_kasparov_equivalence}, 
we can relate the $K$-theory groups to $KK(\C\ell_j,A_b)$ or 
$KKO(C\ell_{j,0},A_b)$. Hence we can 
consider the map
\begin{align} \label{eq:KKO_index_pairing}
 &KKO(C\ell_{j,0},C(\Omega)\rtimes_{\alpha,\theta}\Z^d) \times 
   KKO(C(\Omega)\rtimes_{\alpha,\theta}\Z^d \hat\otimes C\ell_{0,d}, C(\Omega))  
   \to KKO(C\ell_{j,d},C(\Omega))
\end{align}
given by the internal Kasparov product, where the class in the group
$KKO^d(C(\Omega)\rtimes_{\alpha,\theta}\Z^d, C(\Omega))$ is represented by 
the fundamental $K$-cycle for the $\Z^d$-action, $\lambda_b$ from 
Equation \eqref{eq:fund-cl-bulk} (an analogous map 
occurs in the complex case).

Like the case of the quantum Hall effect, where the Hall conductance 
is related to a pairing of the Fermi projection with a cyclic cocycle or 
$K$-homology class, we 
claim that the quantities of interest in topological insulator systems arise as 
pairings/products of this type.

Without specifying a particular $K$-theory class, we can 
make some general comments about the pairing with the fundamental $K$-cycle. 
The unbounded product with the fundamental $K$-cycle has the general form
$$
 \left(C\ell_{j,d}, \, E_{C(\Omega)}, \, 
   \wt{X}, \, \Gamma \right),
$$
with $E_{C(\Omega)}$ a countably generated $C^*$-module with a left-action 
of $C\ell_{j,d}$. The construction of the product is done in such a way 
that the left-action of $C\ell_{j,d}$ graded-commutes with the 
unbounded operator $\wt{X}$. This implies that the topological information 
of interest is contained in the kernel $\Ker(\wt{X})$ as a $C^*$-submodule 
of $E_{C(\Omega)}$ (when this makes sense, see~\cite[Appendix B]{BCR15}).

Let us now associate an analytic index to the product in Equation 
\eqref{eq:KKO_index_pairing}.
\begin{defn}
We let ${}_{r,s}{\mathfrak{M}}_{C(\Omega)}$ be the Grothendieck group of equivalence 
classes of real $\Z_2$-graded right-$C(\Omega)$ $C^*$-modules carrying a 
graded left-representation of $C\ell_{r,s}$.
\end{defn}

Using the notation of Clifford modules, $\Ker(\wt{X})$ determines a class 
in the quotient group ${}_{j,d}{\mathfrak{M}}_{C(\Omega)}/i^*({}_{j+1,d}{\mathfrak{M}}_{C(\Omega)})$, 
where $i^*$ comes from restricting a Clifford action of 
$C\ell_{j+1,d}$ to $C\ell_{j,d}$. Next, we use an 
elementary extension of the 
Atiyah--Bott--Shapiro isomorphism, \cite{ABS64} and \cite[{\S}2.3]{SchroderKTheory}, 
to make the identification
$$  
{}_{j,d}{\mathfrak{M}}_{C(\Omega)}/i^*{}_{j+1,d}{\mathfrak{M}}_{C(\Omega)} 
  \cong  KO_{j-d}(C(\Omega)). 
$$

\begin{defn}
The Clifford index, $\Index_{j-d}(\wt{X})$, of $\wt{X}$ is given by 
the class
$$ [\Ker(\wt{X})] \in 
{}_{j,d}{\mathfrak{M}}_{C(\Omega)}/i^*{}_{j+1,d}{\mathfrak{M}}_{C(\Omega)} 
  \cong  KO_{j-d}(C(\Omega)) $$
\end{defn}

We remark that $\Index_k$ reproduces the usual (real) $C^*$-module Fredholm index 
as studied in~\cite[Chapter 4]{Polaris} if $k=0$, see~\cite{SchroderKTheory}.

Of course we may instead wish to use an invariant probability measure $\bP$ 
on $\Omega$ to obtain invariants averaged over the disorder space rather than 
$K$-theory classes of the disorder. We will return to this question 
in Section \ref{subsec:semifinite_pairings}.

\subsubsection{The bulk-edge correspondence}
Let us now apply Theorem \ref{thm:KKO_factorisation} to our study on 
pairings. Denoting by $[x_b]$ the $K$-theory class represented by 
the symmetry compatible Hamiltonian, we have that 
$$
  [x_b] \hat\otimes_{A_b} [\lambda_b] = 
   (-1)^{d-1} [x_b] \hat\otimes_{A_b} [\mathrm{ext}] 
     \hat\otimes_{A_e} [\lambda_e]
$$
with $\lambda_e$ the fundamental $K$-cycle for the edge algebra.
Using associativity of the Kasparov product to group the terms on the right in two different ways, 
the real index pairing will either be a pairing
\begin{align*}
  &KKO(C\ell_{j,0}, C(\Omega)\rtimes_{\alpha,\theta}\Z^d) \times 
    KKO(C(\Omega)\rtimes_{\alpha,\theta}\Z^d \hat\otimes C\ell_{0,d},\R) 
     \to KO_{j-d}(C(\Omega)), 
\end{align*}
the bulk invariant, or a new (but equivalent) pairing
\begin{align*}  
 &KKO(C\ell_{j,0}\hat\otimes C\ell_{0,1}, C(\Omega)\rtimes_{\alpha^\|,\theta}\Z^{d-1}) \times 
   KKO(C(\Omega)\rtimes_{\alpha^\|,\theta}\Z^{d-1} \hat\otimes C\ell_{0,d-1},\R) \\
     &\hspace{12cm} \to  KO_{j-d}(C(\Omega)). 
\end{align*}
The second pairing yields an invariant that comes from  
$A_e =C(\Omega)\rtimes_{\alpha^\|,\theta}\Z^{d-1}$, the edge algebra of a system with 
boundary. Theorem \ref{thm:KKO_factorisation} 
ensures that regardless of our choice of pairing, the result is the 
same and so we obtain the bulk-edge correspondence for real and complex algebras. 
In particular non-trivial bulk invariants imply non-trivial edge invariants and vice versa.
Furthermore, we see that the bulk-edge correspondence continues to hold under 
the addition of weak disorder.
In complex examples, the 
value of the edge pairing can be interpreted as a response coefficient of the edge system like, for instance, 
the conductance of a current concentrated at the boundary of the sample~\cite{KR06,SBKR02,KSB04b, PSBbook}. 

\begin{remark}[Wider applications of Theorem \ref{thm:KKO_factorisation}] \label{remark:Wider_applications_of_bulkedge}
The bulk-edge correspondence and Theorem \ref{thm:KKO_factorisation} are 
largely independent of the symmetry considerations of topological phases. 
Instead, it is a general property of the (real or complex) unbounded Kasparov module 
representing the short exact sequence
$$  0 \to \calK \otimes C(\Omega)\rtimes_{\alpha^\|,\theta} \Z^{d-1} \to \calT(\alpha_d) \to C(\Omega)\rtimes_{\alpha,\theta} \Z^d \to 0  $$
and the fundamental $K$-cycles on the ideal and quotient algebras we have constructed.

In particular, the fact that the factorisation occurs on the $K$-homological part of the 
index pairing means other $K$-theory classes and symmetry types can be considered without 
changing the result and is independent of the symmetries present. 
For example, if we were to consider symmetry compatible Hamiltonians 
of a group $\tilde{G}$ that included spatial involution or other symmetries, then 
provided that the symmetry data can be associated to a class in 
$KKO(C^*(\tilde{G}),A_b)$ (or complex), 
the pairing with $\lambda_b$ would still display the bulk-edge correspondence.

{\em Separating the topological information arising from the 
internal symmetries of the Hamiltonian from the 
(non-commutative) geometry of the 
Brillouin zone highlights an advantage of using Kasparov theory to study topological 
systems with internal symmetries. We obtain the flexibility to change the 
$K$-theoretic data without affecting the geometric information that is used to obtain the 
topological invariants of interest and vice versa. }
\end{remark}

\subsection{Semifinite spectral triples and non-torsion invariants} \label{subsec:semifinite_pairings}
Given a unital $C^*$-algebra (real or complex) $B$ with possibly twisted $\Z^d$-action, 
we have given a general procedure to obtain an unbounded $B\rtimes_{\alpha,\theta}\Z^d$-$B$ module 
and class in $KKO(B\rtimes_{\alpha,\theta}\Z^d\hat\otimes C\ell_{0,d}, B)$.
If the algebra $B$ has a faithful, semifinite and norm-lower semicontinous trace 
$\tau_B$, there is a general method by which we can obtain a semifinite spectral triple as studied 
in~\cite{KNR,LN04,PR06}. Such a condition on $B$ is satisfied in the physically 
interesting case of $B=C(\Omega,M_N)$, where the disorder space of configurations 
is equipped with a probability measure $\bP$ such that $\mathrm{supp}(\bP)=\Omega$. 
The measure is usually assumed to also be invariant and ergodic under the $\Z^d$-action.

Given the $C^*$-module $\ell^2(\Z^d,B)$ and trace $\tau_B$, we consider the 
inner-product
$$
  \langle \lambda_1\otimes b_1,\lambda_2 \otimes b_2 \rangle = 
    \tau_B\!\left( (\lambda_1\otimes b_1 \mid \lambda_2 \otimes b_2)_B \right) 
    = \langle \lambda_1, \lambda_2 \rangle_{\ell^2(\Z^d)} \, \tau_B(b_1^*b_2),
$$
which defines the Hilbert space $\ell^2(\Z^d)\otimes L^2(B, \tau_B)$ where 
$L^2(B, \tau_B)$ is the GNS space.

\begin{lemma}
The algebra $A=B\rtimes_{\alpha,\theta} \Z^d$ acts on $\ell^2(\Z^d)\otimes L^2(B, \tau_B)$.
\end{lemma}
\begin{proof}
This follows from the identification of $\ell^2(\Z)\otimes L^2(B,\tau_B)$ 
with $\ell^2(\Z,B)\otimes_B L^2(B,\tau_B)$.
\end{proof}

\begin{prop}[\cite{LN04}, Theorem 1.1]
Given $T\in\End_{B}(\ell^2(\Z^d,B))$ with $T\geq 0$, define
$$  \Tr_\tau(T) = \sup_{I} \sum_{\xi\in I}\tau_B\!\left[( \xi\mid T\xi)_{B}\right], $$
where the supremum is taken over all finite subsets $I$ of $\ell^2(\Z,B)$ such that 
$\sum_{\xi\in I}\Theta_{\xi,\xi}\leq 1$. Then $\Tr_\tau$ is a semifinite norm-lower semicontinuous 
trace with the property $\Tr_{\tau}(\Theta_{\xi_1,\xi_2}) = \tau_B[(\xi_2\mid \xi_1)_{B}]$.
\end{prop}

\begin{lemma}
Let $\End_B^{00}(\ell^2(\Z^d,B))$ be the algebra of the span of rank-$1$ operators,
$\Theta_{\xi_1,\xi_2}$ with $\xi_1,\xi_2\in\ell^2(\Z^d,B)$, and $\calN$ be the von Neumann algebra 
$\End_{B}^{00}(\ell^2(\Z^d,B))''$ 
with weak-closure taken in the bounded operators on 
$\ell^2(\Z^d)\otimes L^2(B,\tau_B)$. 
Then the trace $\Tr_\tau$ extends to a trace on the positive cone $\calN_+$.
\end{lemma}
\begin{proof}
This is just the dual trace construction, or an explicit check can be made as in \cite[Proposition 5.11]{PR06}.
\end{proof}

\begin{prop} \label{prop:semifinite_spec_trip}
For $\calA = B_{\alpha,\theta}\Z^d$, the tuple
$$
\bigg(\calA\hat\otimes C\ell_{0,d},\, \ell^2(\Z^d)\otimes L^2(B, \tau_B)\otimes \bigwedge\nolimits^{\!*}\R^d,\, 
   \sum_{j=1}^d X_j\otimes \gamma^j, \,(\calN, \Tr_\tau) \bigg) 
$$
is a $QC^\infty$ and $d$-summable semifinite spectral triple.
\end{prop}
\begin{proof}
The commutators $[X_j,a]$ are bounded by analogous arguments to the 
proof of Proposition \ref{prop:general-toeplitz}.
We use the frame $\{e_m\}_{m\in\Z^d}\subset \ell^2(\Z^d,B)$ with $e_m = \delta_{m}\otimes 1_B$ 
to compute that
\begin{align*}
  \Tr_\tau((1+|X|^2)^{-s/2}) &= \Tr_\tau \!\bigg(\sum_{m\in\Z^d}(1+|m|^2)^{-s/2}\Theta_{e_m,e_m} \bigg) \\
     &= \sum_{m\in\Z^d}(1+|m|^2)^{-s/2}\, \tau_B((e_m\mid e_m)_B) \\
     &=  \sum_{m\in\Z^d}(1+|m|^2)^{-s/2}\,\tau_B(1_B)
\end{align*}
by the properties of $\Tr_\tau$. The sum of $(1+|m|^2)^{-s/2}$ for $m\in\Z^d$ 
will be finite for $s>d$ and $\tau_B(1_B) = 1$. Hence $(1+|X|^2)^{-s/2}$ is 
$\Tr_\tau$-trace-class for $s>d$ as required.

Our spectral triple is $QC^\infty$ if $\calA\hat\otimes C\ell_{0,d}$ preserves the domain of 
$(1+D^2)^{k/2}$ for any $k\in\N$, which reduces to checking that 
$\calA$ preserves the domain of $(1+|X|^2)^{k/2}\otimes 1$ on $\ell^2(\Z^d)\otimes L^2(B,\tau_B)$.
We recall that elements of $\calA$ are of the form $\sum_n S^n b_n$ with $b_n\in B$ and 
$n\in\Z^d$ a multi-index such that $|n|<R$ for some $R$. We compute that
\begin{align*}
  &(1+|X|^2)^{k/2}\sum_{|n|<R}S^n b_n \cdot(\psi(x)\otimes b) \\
  &\hspace{1cm} = \sum_{|n|<R}\! (1+|x-n|^2)^{k/2}\psi(x-n) 
     \otimes \alpha^{-x+n}(\theta(x,n)) \alpha^{-x}(b_n) b
\end{align*}
for $x\in\Z^d$.
Because $n$ is strictly bounded and $\theta(x,n)$ is unitary for all $n,x\in\Z^d$, 
the Hilbert space norm of the sum over finite $n$ of
$(1+|x-n|^2)^{k/2}\psi(x-n) \otimes \alpha^{-x}(\alpha^n(\theta(x,n)) b_n) b$ 
is well-defined and norm-convergent for $\psi \otimes b$ in the domain of $(1+|X|^2)^{k/2}$.
\end{proof}

Let us return to the case of $B=C(\Omega)$, where $C(\Omega)$ is endowed with the 
trace $\calT_\bP$ from the probability measure $\bP$.
The semifinite index pairing with the symmetry $K$-theory class 
$[x_b]\in KO_d(A_b)$ (or complex) is given by the 
composition
\begin{align*}
 &KO_d(C(\Omega)\rtimes_{\alpha,\theta}\Z^d) \times KKO(C(\Omega)\rtimes_{\alpha,\theta}\Z^d \hat\otimes C\ell_{0,d}, C(\Omega)) 
   \to KO_{0}(C(\Omega)) \xrightarrow{(\calT_\bP)_\ast} \R
\end{align*}
Hence the semifinite index pairing of the 
spectral triple of Proposition \ref{prop:semifinite_spec_trip} 
measures disorder-averaged topological invariants. Furthermore 
in the case of complex algebras and modules, the semifinite local 
index formula~\cite{CPRS2,CPRS3} gives us computable expressions for the index pairing 
in terms of traces and derivations. 
We expect similar results to hold for the case of real integer invariants (e.g. 
those that arise from $KO_0(\R)$ or $KO_4(\R)$), though delay a 
more detailed investigation to another place.

\subsection{The Kane--Mele model and torsion invariants}
As an example, we consider $2$-dimensional 
systems with odd time reversal symmetry from~\cite{DNSB14b,KM05b} and 
considered in~\cite[Section 4]{BCR15}. Recall the bulk 
Hamiltonian 
$H_{KM}^\omega = \begin{pmatrix} h_\omega & g \\ g^* & \calC h_\omega \calC \end{pmatrix}$ 
acting on 
$\calH_b = \ell^2(\Z^2)\otimes \C^{2N}$,
with $h_\omega$ a Haldane Hamiltonian, $g$ the Rashba coupling and $\calC$ component-wise complex 
conjugation
such that $g^* = -\calC g\calC$. We take the symmetry group $G=\{1,T\}$ whose time-reversal 
involution is implemented by the operator 
$R_T = \begin{pmatrix} 0 & \calC \\ -\calC & 0 \end{pmatrix}$ on $\calH_b$.
We follow~\cite[Section 4]{BCR15} and construct the projective symmetry 
class $[P_\mu^G]\in KKO(C\ell_{4,0},C(\Omega)\rtimes \Z^2)$. The class
$[P_\mu^G]$ is represented by the Kasparov module
$$ 
  \left( C\ell_{4,0}, P_\mu (A\rtimes G)_A^{\oplus 2}, 0, \Gamma \right),
$$
where $P_\mu$ is the Fermi projection, $(A\rtimes G)_A$ is a $C^*$-module 
given by the completion of $A\rtimes G$ under the conditional expectation 
of the $G$-action on $A$ and $\Gamma$ is the self-adjoint unitary given by the phase 
$H_{KM}^\omega|H_{KM}^\omega|^{-1}$.

For simplicity, we will pair $[P_\mu^G]$ with the spectral triple 
coming from the evaluation map $\mathrm{ev}_\omega$. Namely we have
the real spectral triple
$$ 
 \lambda_b(\omega) =  \bigg(\calA_b \hat\otimes C\ell_{0,2},\, {}_{\pi_\omega}(\ell^2(\Z^2)\otimes\C^{2N})\otimes \bigwedge\nolimits^{\!*}\R^2,\, X_b,\, \gamma_{\bigwedge^*\R^2}\bigg) 
$$
with $X_b = \sum_{j=1}^2 X_j\otimes 1_{2N}\otimes\gamma^j$ and $\calA_b$ 
a dense subalgebra of $C(\Omega)\rtimes\Z^2$ given by matrices 
of elements in $C(\Omega)_\alpha \Z^2$. 

Let us now consider a bulk-edge system. We may link bulk and edge algebras by 
a short exact sequence, which gives rise to a class 
$[\mathrm{ext}]\in KKO(C(\Omega)\rtimes\Z^2 \hat\otimes C\ell_{0,1},C(\Omega)\rtimes\Z)$ 
by the procedure in Section \ref{subsec:extension_module}. We will omit the 
details and instead focus on the bulk and edge pairings, where the edge spectral 
triple is given by
$$
\lambda_e(\omega) = \left( \calA_e \hat\otimes C\ell_{0,1}, 
  {}_{\pi_\omega}(\ell^2(\Z)\otimes \C^{2N}) \otimes \bigwedge\nolimits^{\!*} \R, 
    X_1\otimes\gamma_\mathrm{edge},\gamma_{\bigwedge^*\R}\right)
$$
with $\calA_e$ a dense subalgebra of $C(\Omega)\rtimes\Z$.
Our bulk pairing is the product
\begin{align*}
  &\left[P_\mu^G\right] \hat\otimes_{A}\left[\left(\calA \hat\otimes C\ell_{0,2}, 
  {}_{\pi_\omega}(\ell^2(\Z^2)\otimes\C^{2N})\otimes \bigwedge\nolimits^{\!*}\R^2, X_b, 
    \gamma_{\bigwedge^*\R^2}\right)\right]  \\ 
  &KKO(C\ell_{4,0},A) \times KKO(A\hat\otimes C\ell_{2,0},\R) \to KO_2(\R) \cong \Z_2,
\end{align*}
which can be expressed concretely as the Clifford module valued index,
\begin{align*}
  \left\langle [P_\mu^G], [\lambda_b(\omega)] \right\rangle &= \left[\Ker(P_\mu X_b P_\mu)\right] \\
  &\cong \mathrm{Dim}_\C \Ker[P_\mu((X_1+iX_2)\otimes 1) P_\mu] \,\mathrm{mod}\,2
\end{align*}
using a particular choice of Clifford generators.
By Theorem \ref{thm:KKO_factorisation} and
 the associativity of the Kasparov product, this is the same as the pairing
\begin{align*}
  -&\left( [P_\mu^G] \hat\otimes_{A} [\text{ext}] \right) \hat\otimes_{A_e} 
  \left[\left( \calA_e\hat\otimes C\ell_{0,1}, 
  {}_{\pi_\omega}(\ell^2(\Z)\otimes \C^{2N}) \otimes \bigwedge\nolimits^{\!*} \R, 
    X_1\otimes\gamma_\mathrm{edge},\gamma_{\bigwedge^*\R}\right)\right],
\end{align*}
a map
\begin{align*}
  &KKO(C\ell_{4,0}\hat\otimes C\ell_{0,1}, C(\Omega)\rtimes \Z) \times 
   KKO((C(\Omega)\rtimes \Z)\hat\otimes C\ell_{0,1},\R) 
    \to KO_{4-1-1}(\R) \cong \Z_2, 
\end{align*}
which is now an invariant of the edge algebra $C(\Omega)\rtimes\Z$. Analytic 
expressions can be derived for the edge invariant by taking the Clifford index 
of the (edge) product module.

We would like to examine the edge pairing more closely. We first review what occurs in the 
complex setting as developed in~\cite{KR06,SBKR02}. Let $\Delta\subset \R$ be an 
open interval of $\R$ such that $\mu\in\Delta$ and $\Delta$ is in the complement of the 
spectrum of $H_{KM}^\omega$. By considering states in the
image of the spectral projection $P_\Delta = \chi_\Delta(\Pi_s H_{KM}^\omega \Pi_s)$ 
for $\Pi_s:\ell^2(\Z^2) \to \ell^2(\Z\otimes \{\ldots,s-1,s\})$ the projection, 
we are focusing precisely on eigenstates of the Hamiltonian with edge that do not exist in the 
bulk system, namely edge states. We use the projection $P_\Delta$ to define the unitary
$$ 
U(\Delta) = \exp\!\left(-2\pi i\frac{\Pi_s H_{KM}^\omega \Pi_s - \inf(\Delta)}{\mathrm{Vol}(\Delta)}P_\Delta\right). 
$$
It is a key result of~\cite{KR06,SBKR02} that $U(\Delta)$ is a 
unitary in $(C(\Omega)\rtimes\Z)_\C$ and, furthermore, represents the image of the Fermi projection under 
the (complex) $K$-theory boundary map.
That is, the unitary $[U(\Delta)]\in K_1((C(\Omega)\rtimes\Z)_\C)$ represents the 
complex Kasparov product $[P_\mu]\hat\otimes_{A_\C}[\text{ext}]$ 
for trivially graded algebras. The authors of \cite{SBKR02} show that the pairing of 
$[U(\Delta)]$ with the boundary spectral triple can be expressed as
\begin{align}  
\label{eq:edge_conductance_qH}
 \sigma_e &= -\frac{e^2}{h} \hat\calT\left(U(\Delta)^{*} i[X_1,U(\Delta)]\right) \nonumber \\ 
 &= - \lim_{\Delta\to\mu} \frac{1}{\mathrm{Vol}(\Delta)} \hat\calT(P_\Delta i[X_1,\Pi_s H_{KM}^\omega \Pi_s]), 
\end{align}
where $\hat\calT = \calT\otimes \Tr$ is the trace per unit volume 
along the boundary and operator trace normal to the boundary. 
One recognises Equation \eqref{eq:edge_conductance_qH} 
as measuring the conductance of an edge current (as $P_\Delta$ projects onto edge states). 
Unfortunately, in the Kane--Mele example, the expression 
$\hat\calT(P_\Delta i[X_1,\Pi_s H_{KM}^\omega \Pi_s])$ is zero as there is no net 
current and the cyclic cocycle cannot detect the $\Z_2$-index we associate to the edge channels. 
Under additional symmetries of the bulk Hamiltonian, non-zero cyclic pairings 
for the edge can be 
computed using the spin Chern number (see~\cite{SchulzBaldes13, SchulzBaldesSkewFred}), 
though the resulting indices are no-longer a property of only time-reversal symmetry.
This leads us to summarise the remaining difficulties in our approach.

\section{Open problems}

{\bf Interpreting torsion-valued invariants.}
A concrete representation of the torsion-valued index pairings that give rise to both bulk and 
edge pairings is a much more difficult task than in the complex case, where invariants 
can be expressed as the Fredholm index of the operators of interest. This is because 
we have to consider Kasparov products, which give rise to a `Clifford module' index. 

One advantage of unbounded 
Kasparov theory is that the operators we deal with and the modules we build have geometric 
or physical motivation and so can be linked to the underlying system. In particular 
it would be desirable
to link the Clifford module index with a more physical 
expression for the edge pairing as the edge invariant is meant to be 
directly linked to the existence of edge channels on a system with boundary.
This remains an open problem in the 
field and is related to the difficulty of measuring torsion-labelled states. 
See~\cite{Kellendonk16} for recent progress.

A further complication is that there are two copies of $\Z_2$ in the real $K$-theory of a 
point: in degrees 1 and 2. These different indices arise in different ways, and have 
different interpretations: see \cite{CPS}.

{\bf Disorder, localisation, and spectral gaps.}
Throughout this article we have required that the Hamiltonian retains a spectral gap.
This is so that there is an unambiguous method of constructing $K$-theory classes associated
to the Hamiltonian. It is well-known, however, that the complex pairing with
the class of the fundamental 
$K$-cycle continues to make sense whenever the spectrum of the Hamiltonian
has a {\em gap of extended states}. 

For the quantum Hall effect Bellissard et al. showed that this phenomenon arises 
physically from localisation arising from disorder~\cite{Bellissard94}. 
Mathematically this was
reflected in the properties of commutators between the position operator
$X$ and the projector $\chi_{(-\infty,a]}(H)$
for $a$ a point in such a gap of extended states of $H$. In turn this 
behaviour is seen in the expression for the Chern character, which makes
sense for non-torsion invariants, but seemingly has no analogue for torsion
invariants. Importantly, disorder does not affect the construction
of the  fundamental $K$-cycle.

{\bf The real local index formula.}
For non-torsion invariants of topological insulators, 
it seems reasonable to expect that the entirety of
the machinery developed for complex topological insulators in~\cite{PSBbook} 
has an exact analogue. The use of a Chern character,
the analysis of disorder and localisation, the Kubo formula
relating topological invariants to linear response coefficients all
make perfect sense.

While all this seems quite reasonable, there is a quite a lot of detail to check.
There are many steps in the production of, for instance, the local index
formula, and the heavy reliance on spectral theory means that numerous 
details of the proof require a careful check.

{\bf The $K$-theory class.}
Relating experiments on topological phases to details of
single particle Hamiltonians with symmetries is not a simple
task.  The measured quantities should correspond to pairings
of the fundamental $K$-cycle with $K$-theory classes of the observable
algebra arising from the Hamiltonian. 

Having the $K$-homology data (the fundamental $K$-cycle) fixed allows
us to compare the results of future experiments with the predictions 
arising from pairing with different $K$-theory classes. 



\end{document}